\documentclass[12pt]{article}

\usepackage{natbib}

\usepackage{amsthm}
\usepackage{amsmath}
\usepackage{amssymb}
\usepackage{enumerate}
\usepackage{xcolor}
\usepackage{listings}
\usepackage{multirow}
\usepackage{color}
\usepackage[pdfencoding=auto, psdextra]{hyperref}
\usepackage{doi}
\usepackage{booktabs}
\usepackage{rotating}
\usepackage{bm}
\usepackage{tablefootnote}
\usepackage[normalem]{ulem}
\usepackage{array}
\usepackage{longtable}
\usepackage{bbm}
\usepackage{cleveref}
\usepackage{tabularx}

\newcommand{\E}{\mathbb{E}}
\newcommand{\ed}{\mathrm{d}}
\newcommand{\R}{\mathbb{R}}
\newcommand{\id}{{\bf 1}}
\renewcommand{\P}{\mathbb{P}}
\newcommand{\Q}{\mathbb{Q}}

\newcommand{\EIS}{\mathrm{EIS}}

\newcolumntype{L}[1]{>{\raggedright\let\newline\\\arraybackslash\hspace{0pt}}p{#1}}

\DeclareMathOperator{\Var}{Var}
\DeclareMathOperator{\spow}{sp}

\DeclareMathOperator{\argmax}

\newcommand{\drift}{\mathrm{drift}}
\newcommand{\vol}{\mathrm{vol}}

\newtheorem{theorem}{Theorem}

\newtheorem{corollary}[theorem]{Corollary}
\newtheorem{proposition}[theorem]{Proposition}
\newtheorem{algorithm}[theorem]{Algorithm}

\theoremstyle{definition}
\newtheorem{definition}[theorem]{Definition}
\numberwithin{equation}{section}
\numberwithin{theorem}{section}

\setcitestyle{authoryear,open={(},close={)}}

\begin{document}

\author{\small{John Armstrong\textsuperscript{a}, Cristin Buescu\textsuperscript{b}, James Dalby\textsuperscript{c}}\\
\small{\textsuperscript{a,b,c}Department of Mathematics, King's College London, London, UK}\\
\small{\textsuperscript{a}john.armstrong@kcl.ac.uk, \textsuperscript{b}cristin.buescu@kcl.ac.uk, \textsuperscript{c}james.dalby@kcl.ac.uk} }

\title{Optimal post-retirement investment under longevity risk in collective funds}

\date{}              


\maketitle

\begin{abstract}
We study the optimal investment problem for a homogeneous collective of $n$ individuals investing in a Black-Scholes model subject to longevity risk with Epstein--Zin preferences. 
We compute analytic formulae for the optimal investment strategy, consumption is in discrete-time and there is no systematic longevity risk. We develop a stylised model of systematic longevity risk in continuous time which allows us to also obtain an analytic solution to the optimal investment problem in this case. We numerically solve the same problem using a continuous-time version of the Cairns--Blake--Dowd model.
We apply our results to estimate the potential benefits of pooling longevity risk
over purchasing an insurance product such as an annuity, and to estimate the 
benefits of optimal longevity risk pooling in a small heterogeneous fund.
\end{abstract}


\section{Introduction}

In this paper we calculate the optimal consumption and investment strategies
for a pension investor in a Black--Scholes market subject to
longevity risk, but who has access to longevity insurance through a tontine structure or a similar scheme of longevity credits. We only consider the drawdown stage of retirement.
Our focus is on situations where we can obtain analytic results. 
For this reason, we will throughout model the preferences of the investor using either a power utility, or more generally, Epstein--Zin preferences.

We begin by considering the discrete-time case where there is no systematic longevity risk. This scenario is the simplest and enables us to build the necessary intuition to understand Epstein-Zin utility and its impact on pension problems, before moving onto more complex scenarios in continuous time.  We assume that all investors are identical. We assume a tontine structure where all investors have agreed that each year the funds remaining from any dying investor will be shared evenly among the survivors as longevity credits. Under these circumstances, we will see that it is possible to solve the optimal investment problem to obtain difference equations determining the optimal strategy. We obtain the solution for both finite and infinite funds.

As an application of these results, we can prove rigorously that our infinite fund represents a limit of the finite fund case. We also use our result to show that an effective strategy for managing a heterogeneous fund where investors have different preferences and mortality distributions, is for each investor to trade and consume as though they were in a homogeneous fund. When combined with a longevity credit mechanism that accounts for the different market risks taken by different individuals, we find numerically that this results in strategies that perform close to optimally. This demonstrates that results obtained on the assumption of homogeneous funds are in fact useful in practice.


We next introduce systematic longevity risk. In order to obtain tractable problems, in this section we work in continuous time. In the first instance, we are interested primarily in analytic tractability, so we introduce a new stochastic model for longevity risk which is designed to be such. 
The cost of this tractability is that the model cannot be calibrated to realistic data, but can be used to obtain qualitative results on behaviour by explicitly solving the Hamilton-Jacobi-Bellman (HJB) equation. 

Finally, we consider a realistically calibrated 1-factor model for systematic longevity risk based on the Cairns--Blake--Dowd (CBD) model \cite{CBD_article}. We must solve the resulting HJB equation using PDE methods, but the use of homogeneous preferences at least allows us to reduce the dimensionality of the problem. This model enables us to study the consumption-investment problem with more realistic mortality risk. These results confirm are earlier qualitative results.

Our results demonstrate that, as expected, large benefits are available through longevity-credit mechanisms and from incorporating investment in risky assets post retirement. These benefits can be obtained
entirely through a mechanism of mutual insurance and do not
require any inter-generational transfers of wealth, as seen in Collective Defined Contribution Schemes.

We also show that when individuals face systematic longevity risk, it is sometimes possible to obtain better pension outcomes. This is because in a model for systematic longevity risk, one gains increasing information about one's mortality with time and one can potentially use this information to make better investment and consumption decisions. This provides an additional means by which longevity pooling in a collective fund can out perform an annuity than those identified in \cite{bear_or_insure} for instance.

\medskip

Our results contribute to the extensive literature
on optimal investment that was initiated by Merton in \cite{merton1969lifetime}. The results in our first section provide an improvement on the similar results of \cite{stamos} by considering finite funds of arbitrary size and incorporating Epstein--Zin utility. Epstein--Zin utility allows one to have preferences for consumption over time where one can separately consider risk-aversion and the diminishing returns from increasing consumption at one moment in  time. It has been shown that this allows one to resolve a number of asset pricing puzzles \citep{bansalYaron, bansal, benzoniEtAl, bhamraEtAl}. It therefore seems likely that any realistic model of an individual's preferences for pension investments will also separate these factors. 
The tractability of optimal investment problems with Epstein--Zin utility is studied extensively in \cite{campbellViceira}.
Fully rigorous treatments of continuous-time Epstein--Zin investment problems without mortality, can be found in \cite{optimal_EZ_finite}
and \cite{xing}.

Our paper also contributes to the growing literature on tontines. An extensive review of the history of tontines, their potential benefits and of the more recent literature can be found in \cite{milevsky}.
Other authors have also identified optimal investment strategies for homogeneous pools of investors subject to various constraints on the design of the pension product. \cite{milevsky2015} considers the optimal investment problem for a tontine invested in a bond. \cite{chenAndRach} considers a tontine structure with a minimum guaranteed payout. \cite{chenRachSehner2020} considers structures combining annutities and tontines. \cite{boadoPenas} consider a product where an individual and shared fund are maintained according to a specific set of rules that ensure the funds are kept within certain "corridors".

\medskip

Our results are intended to be simple, tractable and used as a benchmark for studying risk-pooling
mechanisms in collective pensions. Our ultimate conclusions suggest that tractable homogeneous preferences such as power utility and Epstein--Zin preferences, provide a reasonable model for individuals who have insufficient funds to achieve a fully adequate pension in retirement, but do not yield plausible strategies for individuals with large pension pots. As a result, we believe that inhomogeneous preference models are required for realistic modelling.

\section{Pooling idiosyncratic longevity risk}\label{sec:discrete_model}

For our market model, we take the Black--Scholes model with a risk-free asset with
interest rate $r$ and a single risky asset $S_t$, satisfying the SDE
\begin{equation}
\ed S_t = S_t(\mu \, \ed t + \sigma \, \ed W_t), \qquad S_{t_0}\in\mathbb{R}^+,
\label{eq:blackScholesMerton}
\end{equation}
for a constant drift $\mu$ and volatility $\sigma$, and a $1$-dimensional Brownian motion $W_t$.
We assume that consumption occurs in discrete-time at evenly-spaced time
points ${\cal T}=\{0,\delta t, 2\, \delta t, \ldots, T \}$.

In this section, we will be interested only in idiosyncratic longevity risk, so we will assume individuals have independent mortality, with the unconditional probability 
of dying in the interval $(t, t+\delta t]$ being $p_t$. Hence, the survival
probability $s_t$ between times $t$ and $t+\delta t$, is given by:
\begin{equation}
s_t =
\frac{\sum_{i=1}^\infty p_{t + i \delta t}}{\sum_{i=0}^\infty p_{t + i \delta t}}.
\label{eq:defs}
\end{equation}
When an individual dies, their remaining wealth is shared evenly among the survivors.

An individual's consumption of their pension is given by a stochastic process $(c_t)_{t \in {\cal T}}$. Let $\tau$ be a ${\cal T}$-valued random variable denoting the last time at which the individual was able to consume. Note that $\tau$ is not a stopping time, but $\tau + \delta t$ is.

We will assume that each individual seeks to maximize a value function given
by a homogeneous Epstein--Zin utility which we now define.
\begin{definition}
	\label{def:homogeneousEpsteinZin}	
	{\em Discrete-time homogeneous Epstein--Zin utility with mortality} depends on parameters $\alpha \in (-\infty,1) \setminus\{0\}$, $\rho \in (-\infty,1) \setminus\{0\}$,
	and $0<\beta \leq 1$.
	It is the $\R_{\geq 0} \cup \{ \infty \}$-valued stochastic process
	defined recursively by
	\begin{equation}
	Z_t(c, \tau) =
	\begin{cases}
	0 & t > \tau; \rho>0 \\
	\infty & t > \tau; \rho<0 \\
	\left[ c_t^\rho + \beta \, \E_t( Z_{t+\delta t}(c, \tau)^\alpha )^\frac{\rho}{\alpha} \right]^\frac{1}{\rho} & \text{otherwise}.
	\end{cases}
	\label{eq:defepsteinzin}
	\end{equation}
	To interpret this formula we use the convention $\infty^a=0$ for $a<0$.
\end{definition}
We call these preferences homogeneous because they satisfy the equation $Z_t(\zeta c, \tau)= \zeta Z_t(c, \tau)$ for any constant $\zeta$. Our choice of value for the utility when $t > \tau$, is determined by the requirement that positive homogeneity still holds.

The parameter $\rho$ can be interpreted as controlling satiation, as the term $c^\rho_t$ will result in diminishing returns to consumption at any instant in time. This will lead optimal investors, to spread their consumption over time, and so the choice of $\rho$ can also be viewed as determining intertemporal substitution.
The parameter $\alpha$ is a risk-aversion parameter. The parameter $\beta$ is a discount rate indicating the level of preference for earlier consumption. When $\alpha=\rho$, maximizing these Epstein--Zin
preferences is equivalent to maxmizing the von Neumann--Morgenstern utility:
\[
\E\left( \sum_{t \in {\cal T}, t \leq \tau } c_t^\alpha \right).
\]
We will primarily be interested in the case when $\beta=1$. This is because we have included mortality in our model which already provides a motivation for preferring earlier consumption and additional discounting does not seem appropriate when considering a pension income. 

If we allow continuous-time trading between consumption, then to find
the optimal investment strategy we must solve a sequence of one-period
investment problems for each discrete-time step.
The value function in the resulting HJB equation will depend upon the variables $(t, S_t, w_t)$ where $w_t$ is the total wealth of the individual at time $t$.
Because Epstein--Zin preferences are positively homogeneous, each one-period
investment problem will be equivalent to a classical Merton problem with a power
utility, and so can be solved analytically. This allows us to obtain
difference equations for the optimal investment strategy. We state the result
and give the detailed calculations in Appendix \ref{appendix:longevitypooling}.

\begin{theorem}
\label{thm:ezIdioscynractic}
Let $z_{n,t}$ denote the optimal Epstein--Zin utility, \eqref{eq:defValueFunction}, for a collective of $n$ individuals investing an amount $1$ at time $t$. Suppose the discounting
parameter $\beta$ in the preferences is equal to $1$.
The collective is
allowed to invest in the Black--Scholes--Merton market \eqref{eq:blackScholesMerton} in continuous time, but consumption occurs at discrete times $\{0,\delta t, 2\, \delta t, \ldots \}$.

The optimal Epstein--Zin utility is determined by the equations:
\begin{equation}
S_t(n,i):= \binom{n}{i} (s_t)^i (1-s_t)^{n-i}.
\label{eq:defS}
\end{equation}
\begin{equation}
a^*:= \frac{\mu-r}{(1-\alpha) \sigma^2},
\label{eq:defastar}
\end{equation}
\begin{equation}
\xi := a^*(\mu - r) + r - \frac{1}{2}(a^*)^2(1-\alpha)\sigma^2,
\label{eq:defxi}
\end{equation}
\begin{equation}
\tilde{\theta}_{n,t} :=
\begin{cases}
\beta^\frac{1}{\rho} \exp( \xi \, \delta t) \left(
\sum_{i=1}^n \left( \frac{i}{n} \right)^{1-\alpha} S_t(n,i) z_{i,t+\delta t}^{\alpha}
\right)^\frac{1}{\alpha} & n \neq \infty, \\
\beta^\frac{1}{\rho} \exp(\xi \delta t) s_t^{(\frac{1}{\alpha}-1)} z_{\infty, t+\delta t}
& n= \infty
\end{cases}
\label{eq:deftildetheta}
\end{equation}

\begin{equation}
z_{n,t}^\frac{\rho}{1-\rho} = 1 + \tilde{\theta}_{n,t}^\frac{\rho}{1-\rho}.\label{eq:valueFunction4Tilde}
\end{equation}
To compute $z_{n,t}$, note that the last two equations determine $z_{n,t}$ from ${z_{i,t+\delta t}}$, and we know the terminal utility $z_{n,T}$.

The term $a^*$ defined in \eqref{eq:defastar} is the optimal proportion invested in the risky asset. It is independent of both $\rho$ and the distribution of mortality.

The optimal proportion of wealth consumed by each survivor at
a given time point is
\begin{equation}
c^*_t = z_t^\frac{\rho}{\rho-1}
\label{eq:cStarConclusion}.
\end{equation}
\end{theorem}
The term $S_t(n,i)$, defined in \eqref{eq:defS}, is the transition 
probability
from $n$ surviving individuals to $i$ surviving individuals over the interval $(t,t+\delta t]$.

\medskip

Our analytic formulae make it relatively simple
to quantify the convergence of the value function
as the number of individuals tends to infinity.

\begin{theorem}
	\label{thm:cvgcEpsteinZin}	
	Let $z_{\infty,t}$ denote the maximum Epstein--Zin utility at time $t$
	for the infinitely collectivized case, then
	\[
	z_{n,t} = z_{\infty,t} + O(n^{-\frac{1}{2}}).
	\]
\end{theorem}

See Appendix \ref{appendix:longevitypoolingconvergence} for the proof.

Because we are, in effect, solving a sequence of Merton problems it
is possible to compute the probability distribution for the consumption analytically. We state the result in the most interesting cases when $n=1$ and $n=\infty$. To distinguish these cases we define
a constant $C$ by
\begin{equation}
\label{eq:defC}
C =
\begin{cases}
0 & n=1 \\
1 & n=\infty.
\end{cases}
\end{equation}

\begin{theorem}
	Under the same conditions as \ref{thm:ezIdioscynractic},
    and with $n=1$ or $n=\infty$,
	the optimal fund value per survivor at time $t$, $X_t$, follows a log
	normal distribution. Write $\mu^X_t$ and $\sigma^X_t$ for
	the mean and standard deviation of $\log X_t$ so that
	\begin{equation}
	\log X_t \sim N( \mu^X_t, (\sigma^X_t)^2 ).
	\label{eq:distLogWealth}
	\end{equation}
	The standard deviation is given by
	\begin{equation}
	\sigma^X_t = \sigma a^* \sqrt{t},
	\label{eq:sdLogWealth}
	\end{equation}
	where $a^*$ is given by \eqref{eq:defastar}.
	The mean satisfies the difference equation
	\begin{equation}
	\mu^X_{t+\delta t}=\mu^X_{t} + \log(s_t^{-C}) + 
	\log\left( 1 - z_t^\frac{\rho}{\rho-1} \right) + \tilde{\xi} \delta t, \qquad \mu^X_{0} = \log (x_0),
	\label{eq:meanLogWealth}
	\end{equation}
	where $z_t$ is given by \eqref{eq:valueFunction4} and we define $\tilde{\xi}$ by the same formula
	used to define $\xi$, but with $\alpha$ set to zero, i.e.\
	\begin{equation}
	\tilde{\xi} := a^*(\mu - r) + r - \frac{1}{2}(a^*)^2 \sigma^2.
	\label{eq:defTildeXi}
	\end{equation}	
	The optimal consumption per survivor, $c_t$, at time $t$, is also
	log normally distributed with
	\begin{equation}
	\log c_t \sim N( \tfrac{\rho}{\rho-1} \log(z_t) + \mu^X_t,
	(\sigma^X_t)^2 ).
	\label{eq:distConsumption}
	\end{equation}
	The mean of the log consumption per survivor satisfies the equation
	\begin{equation}
	\E( \log c_{t+\delta t} \mid c_t ) = 
	\log(c_t) + \log(s_t^{-C}) + \frac{\rho}{1-\rho} \log(\phi_t) + \tilde{\xi} \delta t
	\label{eq:meanLogConsumptionDynamics}
	\end{equation}
	where $\phi_t$ is given by equation \eqref{eq:defphi}.
	\label{thm:epsteinZinConsumption}
\end{theorem}

The proof is given in Appendix \ref{appendix:longevitypooling}.




To understand the role of Epstein-Zin preferences in Theorem \ref{thm:epsteinZinConsumption}, we specialize to the case of a market where $r=\mu=0$ and 
to preferences with $\beta=1$. This represents the problem of consuming a
fixed lump sum over time when there is no inflation but also no investment opportunities. While not financially reasonable, this problem highlights
how longevity risk affects consumption, when considered in isolation from market risk.
In this case, $c_t$ is a deterministic function. We find from equations \eqref{eq:meanLogConsumptionDynamics}
and \eqref{eq:defphi} that
\begin{equation}
c_{t+\delta t}=
\left(s_t^{\frac{1}{\alpha }-\frac{C}{\rho }}\right)^{\frac{\rho }{1-\rho }} c_t.
\label{eq:consumptionDynamics}
\end{equation}
We note that $s_t$ is a non-zero probability, so $0 < s_t \leq 1$. We may use equation \eqref{eq:consumptionDynamics} to compute whether consumption increases or decreasing over time. We summarize the results in Table \ref{table:direction}.

\begin{table}
	\begin{center}
		\begin{tabular}{ccc} \toprule
			& Collectivised & Individual \\ \midrule
			$\alpha<\rho<0<1$ & Increasing & Decreasing \\
			$\alpha<0<\rho<1$ & Increasing & Increasing \\ 
			$0<\alpha<\rho<1$ & Decreasing & Decreasing \\ 
			$\alpha=\rho<0<1$ & Constant & Decreasing \\
			$0<\alpha=\rho<1$ & Constant & Decreasing \\ \bottomrule
		\end{tabular}
		\caption{The behaviour of consumption with time when $\mu=r=0$, $\beta=1$ and $\alpha\leq\rho$.}
	\end{center}
	\label{table:direction}
\end{table}

Perhaps surprisingly, we find that sometimes consumption
increases over time rather than decreases. When $\alpha>0$, living longer always increases the total lifetime utility. In this case
the risk an individual faces is the risk of dying young with leftover funds, which is addressed by consuming early. When $\alpha<0$, living longer decreases the total lifetime utility, and the lower the income the more sharply it decreases. In this case the risk an individual faces is the risk of living longer and running out of funds, so the individual chooses to save (and postpone consumption) to compensate themselves for this risk.

Some may believe that increased longevity should always increase
utility, in which case only Epstein--Zin utilities with $\alpha>0$ should be considered. Others may believe that it is possible that if one were too poor in old age, living longer might indeed be detrimental. Epstein--Zin utilities with $\alpha>0$ could then be used to model individuals with an adequate pension and $\alpha<0$, to model poor individuals who are unable to achieve an adequate income in retirement. We believe that a
 more realistic model should include an adequacy level such that consumption above the adequacy level results in improved utility with longer lifetimes and consumption below the adequacy results in decreasing utility. 
However, doing so will inevitably break the positive homogeneity of the preferences, leading to an analytically intractable result. Thus, while Epstein--Zin preferences provide an instructive model due to their tractability, non-homogenous
preferences will ultimately be required in a realistic preference model. We will consider preference models which allow for more refined modelling of adequacy in future research.

In the individual case ($n=1$), the concern that one will die young is much more serious. This is why for the individual problem, the fear of an inadequate pension only dominates when both $\alpha<0$ and $0<\rho<1$.

The case $\alpha=\rho$ corresponds to standard
von-Neumann Morgernstern preferences, in which case constant consumption
is optimal.
More significantly, our result also shows the converse. Constant consumption from one period to the next is only optimal if and only if either (i) the survival probability is one, or (ii) $\alpha=\rho$ so one is satisfaction risk-neutral. Hence, even ignoring market effects, constant consumption will be sub-optimal for any realistic parameter choices.

We have not shown the case $\alpha>\rho$ in Table \ref{table:direction} as we found the resulting behaviour to be difficult to interpret as rational, cautious (as understood intuitively) strategies.
In our view, this is because when $\alpha>\rho$, Epstein--Zin preferences do not correctly operationalise the intuitive notion of risk-aversion.
To see why, it is helpful to rewrite the Epstein--Zin utility function.
We
define the signed power function by
\[
\spow_\gamma(x) = \begin{cases}
x^\gamma & \text{ when } \gamma>0 \\
-x^\gamma & \text{ when } \gamma<0
\end{cases}
\]
and define the {\em Epstein--Zin satisfaction}, $z_t$, by
\[
z_t = \spow_{\rho}( Z_t, \rho ).
\]
We may then write the defining equations of homogeneous Epstein--Zin preferences as follows
\begin{equation}
z_t = c_t^\rho + \beta \spow_{\frac{\alpha}{\rho}}^{-1} \left( \E_t\left( \spow_{\frac{\alpha}{\rho}}( z_{t+\delta t} ) \right) \right).
\label{eq:epsteinZinRewritten}
\end{equation}
For deterministic cashflows, these preferences simplify to
\begin{equation*}
z_t = c_t^\rho + \beta z_{t+\delta t}
= \sum_{i=0}^\infty \beta^i c_{i \delta t}^\rho. 
\end{equation*}
Since $\rho$ affects ones preferences even over deterministic cashflows, we interpret this equation as showing that $\rho$ is a parameter measuring satiation, that is the diminishing returns from increased consumption over one time period, rather than risk-aversion.

Equation \eqref{eq:epsteinZinRewritten} shows that if $\alpha>\rho$, the preferences are risk-seeking in the satisfaction, risk-neutral if $\alpha=\rho$ and risk-averse if $\alpha<\rho$. Constant consumption is optimal in a market with $\mu=r=0$ only if $\alpha=\rho$, that is only if one is satisfaction-risk-neutral. For Epstein--Zin preferences, a risk-averse investor will prefer early consumption to a steady income in retirement. This may be one of the reasons many retirees see annuities as unattractive. 

\medskip

It is also interesting to calculate how consumption changes 
according to the available investment opportunities.
If interest rates increase, one may choose to defer consumption to
a later date to benefit from the increased rate.
To quantify this behaviour, one wishes to calculate
the elasticity of intertemporal substitution (EIS) which
is defined as follows.

\begin{definition}
	The {\em elasticity of intertemporal substitution} at time $t$ is defined to be
	\[
	\EIS_t:= \frac{1}{\delta t}
	\frac{\ed}{\ed r} \left( \E( \log( c_{t+1}) )
	- \log( c_{t}) \right).
	\]
	When $c_t$ is deterministic, this definition corresponds
	with the standard definition of \cite{hall}.
\end{definition}

Theorem \ref{thm:epsteinZinConsumption} allows us
to calculate this elasticity.
\begin{corollary}
	\label{corollary:eis}	
	For the optimal investment strategy of Theorem \eqref{thm:ezIdioscynractic}
	we have
	\[
	\EIS_t = \frac{1}{1 - \rho} \left( 1 - \frac{(\mu - r) (1 + \alpha (\rho - 2))}{(\alpha - 1)^2 \sigma^2}
	\right).
	\]
	If $\mu=r$ this simplifies to
	\[
	\EIS_t = \frac{1}{1 - \rho}.
	\]
	In the case of von Neumann-Morgernstern utility we have
	\[
	\EIS_t = \frac{1}{1 - \rho} \left( 1 - \frac{\mu - r}{\sigma^2} \right).
	\]
\end{corollary}
\begin{proof}
	From  \eqref{eq:meanLogConsumptionDynamics} and \eqref{eq:defphi} we immediately
	find
	\[
	\EIS_t = \frac{\ed }{\ed r} \left(\frac{\rho}{1-\rho} \xi - \tilde{\xi} \right).
	\]
	The result is now a simple calculation from \eqref{eq:defastar} and \eqref{eq:defTildeXi}.
\end{proof}

This result explains why we call $\rho$ a satiation parameter rather than following the convention in the economics literature of saying that it is a parameter for EIS. The concept of satiation is a feature inherent to the preferences alone, whereas EIS depends upon other features of the model.

\medskip

We now simulate some example pension outcomes using \Cref{thm:ezIdioscynractic}. In Figure \ref{fig:epstein-zin-with-riskaversion}, we plot the optimal consumption calculated
using homogeneous Epstein--Zin preferences with
mortality distribution given by  the \cite{cmi2018} model
CMI\_2018\_F [1.5\%],
with market parameters
$r=0.027$, $\mu=0.062$, $\sigma=0.15$ and an initial fund of
value $\pounds 126,636$. The parameters are chosen to represent a median-earning female retiring in the UK in 2019. \label{page:marketParams}

Regarding values of $\alpha$ and $\rho$, empirically relevant values are those such that $\alpha<0,\rho>0$ \cite{xing}, at least for incomplete markets. Although, values of $\alpha\in(-7, -1)$ and $\rho\in(-4,1/2)$ are considered in \cite{xing}. The case $\alpha,\rho<0$ will yield relevant results. \cite{havranek2015} finds a mean value of 0.5 in a meta-analysis of empirical studies on the elasticity of inter-temporal substitution. Using the relation we found in Corollary \ref{corollary:eis}, $\rho=-1$ seems a reasonable choice. If anything, this should underestimate the benefits of our scheme when compared to the case $\rho>0$.  While the case $\alpha,\rho>0$ is not typically studied in the literature, our analysis shows it is a qualitatively distinct and interesting case corresponding to individuals with adequate pensions. 

\Cref{fig:epstein-zin-with-riskaversion} shows consumption strategies in the presence of market risk. For $\alpha>0$ (\Cref{fig:epstein-zin-with-riskaversion} (a)), living longer still increases total utility, but consumption increases then decreases, since it is initially postponed to take advantage of investment opportunities. For $\alpha<0$ (\Cref{fig:epstein-zin-with-riskaversion} (b) and (c)), living longer still decreases utility, so consumption continues to be postponed to avoid running out of funds. Additionally, \Cref{fig:epstein-zin-with-riskaversion} visually illustrates how our homogeneous preferences effectively set an adequacy level of zero when $\alpha>0$ and infinity when $\alpha<0$. Because the adequacy level is zero when $\alpha>0$, the investor is willing to accept some years where the pension is unrealistically low (even zero) and this is what leads to the unrealistic consumption strategies seen in Figure \ref{fig:epstein-zin-with-riskaversion} (a). In \Cref{fig:epstein-zin-with-riskaversion} (c), the investor aggressively targets achieving an adequacy of infinity at the expense of almost no consumption until they are close to death. This is also an unrealistic strategy. \Cref{fig:epstein-zin-with-riskaversion} (b) represents a more plausible strategy, with the median scenarios outperforming an annuity at all ages.

We believe investors with sufficiently large pension pots will be concerned more by the risk of dying young, than they will be of running out of money and so will not be adequately modelled by preferences with $\alpha<0$. 
To model such investors, it seems necessary to go beyond homogeneous preferences and include an explicit adequacy level within the model. 
It also seems likely that utility decreasing with increasing lifetime is a necessary feature of any optimization model that will yield reasonable consumption strategies.

\begin{figure}[htb]
    \centering
    \begin{minipage}{0.49\textwidth}
        \centering
	\includegraphics[width=1.0\textwidth]{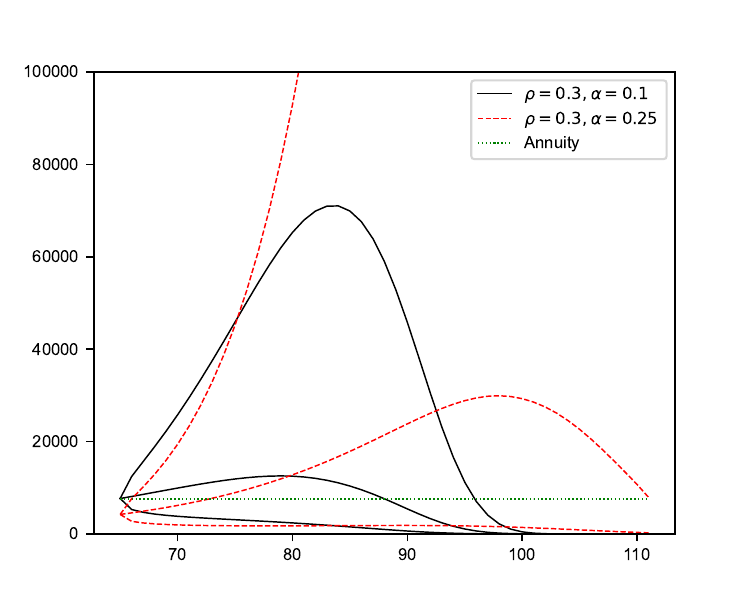}\\
    \textrm{(a)}
    \end{minipage}
    \begin{minipage}{0.49\textwidth}
        \centering
	\includegraphics[width=1.0\textwidth]{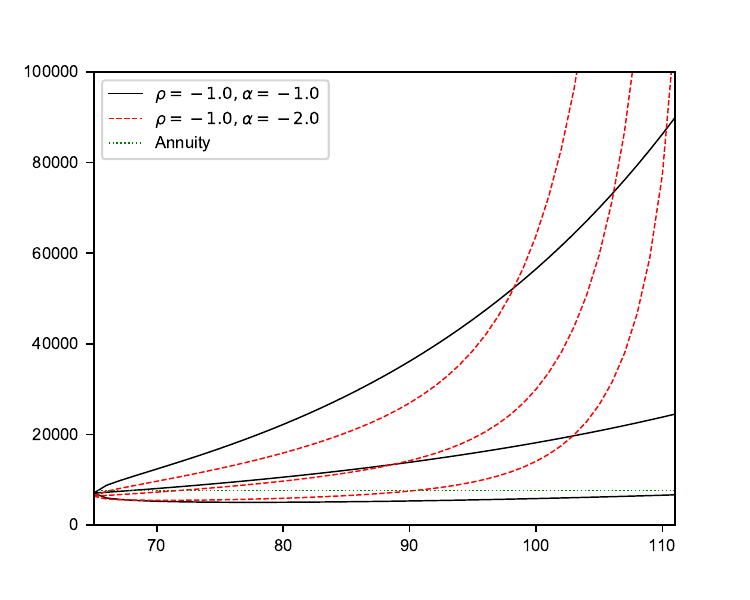}\\
    \textrm{(b)}
    \end{minipage}
    \begin{minipage}{0.49\textwidth}
        \centering
	\includegraphics[width=1.0\textwidth]{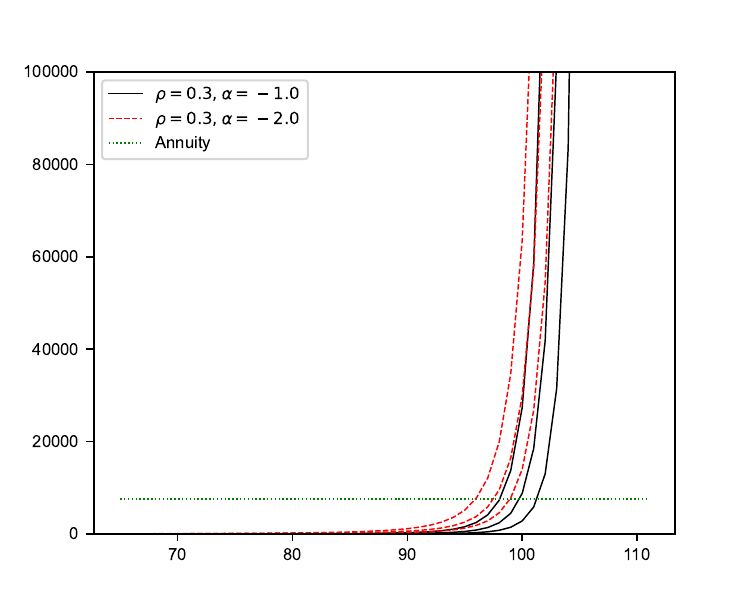}\\
    \textrm{(c)}
    \end{minipage}
	\caption{Fan diagrams of optimal consumption for an infinite collective ($n=\infty$) with homogeneous Epstein--Zin preferences. We illustrate both intertemporally additive von Neumannn--Morgenstern preferences ($\alpha=\rho$) and of satiation-risk-aversion ($\alpha<\rho$). The percentiles in the fans are $(5\%, 50\%, 95\%)$. }
	\label{fig:epstein-zin-with-riskaversion}
\end{figure}

\subsection{Heterogeneous Funds}
\label{sec:heterogeneousNumerics}

As an application of our computation for the optimal investment strategy for a finite number of individuals, we propose an algorithm
for managing heterogeneous funds of individuals with different preferences and mortality distributions.

\begin{algorithm}[Heterogeneous fund algorithm]
	\label{algo:heterogeneous}
	Choose a positive integer $n_{\max}$. Let $n_t$ denote the number of survivors at time $t$. Define $n^\prime_t$ by
	\[
	n_t^\prime = \begin{cases}
	n_t & n_t \leq n_{\max} \\
	\infty & \text{otherwise}.
	\end{cases}
	\]
	Compute the consumption and investment strategy as follows:
	\begin{enumerate}
		\item  Keep accounts of the current funds associated with
		each individual.
		\item At time $t$, for each surviving individual $i$ in the fund, invest and consume
		according to the optimal strategy for a homogeneous fund of $n_t^\prime$
		investors identical to that individual with budget given by their current
		funds. Even if, after consumption, the individual dies at time $t$,
		one should pursue the
		same investment strategy as one would have done if they had survived.
		
		Compute the resulting wealth $\mathring{X}^i_{t+\delta t}$ of each individual
		$i$ who was alive at time $t$.
		\item For an investor $i$ who survives to time $t+\delta t$, we define their
		``contribution'' to the collective at time $t+\delta t$, $\Gamma^i_{t+\delta t}$,
		by
		\[
		\Gamma^i_{t+\delta t}
		= (1-s^i_t) \mathring{X}^i_{t+\delta t}
		\]
		where $s^i_t$ is the survival property of individual $i$ from time $t$ to $t+\delta t$.
		\item When an individual dies, divide their funds among the survivors
		in proportion to each survivor's contribution $\Gamma^i_{t+\delta t}$.	      
	\end{enumerate}
\end{algorithm}

The purpose of the cut-off $n_{\max}$ in this algorithm is simply that it
is computationally expensive to compute the optimal strategy for a collective of $n$ investors, if $n$ is large.

The logic behind this algorithm is that we assume we can divide
our population into large groups of similar individuals.
Let one such group of individuals be close to one particular individual
which we label $\zeta$.
Since the individuals are similar we assume that the optimal
strategies for each member of the group will be similar. If the number
of individuals in the group is large, the optimal strategy for $n_t$ individuals
of a given type will be similar to the optimal strategy for $\infty$ individuals
of a given type. The number of survivors
will also be close to the expected value. Thus, (2) and (3)
together will ensure that the utility achieved by individuals of type close
to $\zeta$, will be close to the utility that can be obtained by an infinite
collective of individuals all of type exactly $\zeta$.

This argument is essentially a compactness and
continuity argument. One could therefore write out a formal topological
proof of its convergence as the number of individuals tends to infinity,
under suitable assumptions, and with small changes to the algorithm
similar to those used to prove Theorem \ref{thm:cvgcEpsteinZin}.
However, we do not believe doing so
would be particularly illuminating.

We note that our method of redistributing the funds of deceased members
could be improved as it is not strictly ``fair''. Let $\Q^M$ denote the pricing measure for the market and $\P^L$ the measure for individual lifetimes. One might wish to choose the specifics of the redistribution mechanism so that the investment of each individual matches the discounted $\Q^M \times \P^L$ expectation of their funds at the start of the next period. Our algorithm only achieves this approximately. See \citep{piggott,sabin2010} to see how this can be fixed, if desired.

We now test the effectiveness of Algorithm \ref{algo:heterogeneous} for a small heterogeneous fund.
We generate a random fund
of $n=100$ individuals. Each individual has inter-temporally additive von Neumann--Morgenstern 
preferences given by a power utility, with the power uniformly generated
in the range $[-1.5,-0.5]$.
The initial wealth of each individual
is taken to be uniformly distributed in the range $[0.5,1.5]$.
The retirement age of each individual is taken to be a uniformly
distributed random integer in the range $60$--$69$. The sex
of the individual is chosen as male with probability $50\%$. All of these
random choices are made independently.
The mortality distribution for each individual
is then computed using the CMI model with longterm rate $1.5\%$
and retirement year of $2019$. The market parameters are chosen as on page \pageref{page:marketParams}.

Using this population, we run $10^6$ simulations of 
the above algorithm with $n_{\max}$ set to $50$. This allows us to compute a sample
average utility, $u^i_S$, for each individual $i$.  We use von Neumann--Morgenstern preferences for our simulations, as it is much simpler to compute expected utilities from a Monte Carlo simulation in this case, than the case of Epstein--Zin preferences.
We define $u^i_n$, to be the expected utility that would be achieved by
individual $i$, if they were to invest in a homogeneous collective of $n$ 
individuals. We define the {\em optimality ratio} for individual $i$ by 
\[
\text{OR}_i := \frac{u^i_S-u^i_1}{u^i_\infty-u^i_1}.
\] 
If this ratio is close to $1$, then the utility experienced by individual $i$
is close to the optimum value they can expect from a
heterogeneous fund.

In Figure, \ref{fig:heterogeneous-histogram} we plot a histogram of the optimality ratio, $\text{OR}_i$, for each of our $100$
fund members. In our example,
the optimality ratio is almost always above $98\%$. This demonstrates
that, even with as few as $100$ investors, our investment strategy is close to the optimal value for a heterogeneous fund.

\begin{figure}
	\centering
	\includegraphics{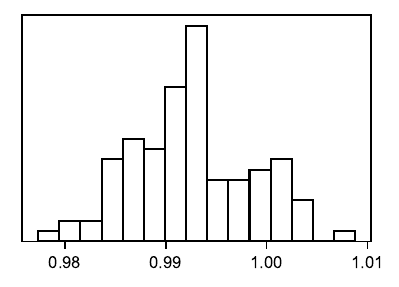}
	\caption{Histogram of the optimality ratio, ${OR}_i$, 
    obtained for a
	randomly generated fund of $100$ investors.}
	\label{fig:heterogeneous-histogram}
\end{figure}

\section{Systematic longevity risk}

Having studied the effects of idiosyncratic longevity risk, we now move on to systematic longevity risk. 
To obtain a fully tractable model for systematic longevity
risk, we need to allow consumption in continuous time and
therefore consider the continuous time version of
Epstein--Zin preferences with mortality.

Let $\lambda_t$ be the force of mortality. We define
the Epstein--Zin aggregator with mortality by
\begin{equation}
f(c,V,\lambda):=\frac{1}{\rho}c^\rho \left( \alpha V \right)^{1-\frac{\rho}{\alpha}} - \left( \frac{\alpha}{\rho} \delta + \lambda_t \right) V.
\label{eq:general_EZ_aggregator}
\end{equation}
Here, the constant $\delta \geq 0$ represents a rate of discounting corresponding to the parameter $\beta$ in the definition of discrete-time Epstein--Zin preferences with mortality.
We define the gain function for Epstein--Zin oreferences with mortality to be
the solution of the Backwards Stochastic Differential Equation 
\begin{equation}
d V_t = -f(c_t, V_t, \lambda_t) \, \ed t + Z_t \ed W_t,
\label{eq:ezBSDE}
\end{equation}
with terminal condition $\lim_{t\to\infty} V_t=0$.
We justify this definition heuristically in Appendix \ref{sec:ezAggregatorMotivation}, arguing that $V_t$
defined in this way is the continuous-time analogue
of the discrete-time quantity $\frac{1}{\alpha} Z_t^{\alpha}$.
The argument is essentially identical to the standard justification
for the definition of continuous-time Epstein--Zin preferences
without mortality, in terms of a BSDE \cite{optimal_EZ_finite}.

We will consider the optimal investment problem for
an insured drawdown fund in the Black-Scholes model, where
the force of mortality satisfies a diffusion equation of the form
\begin{equation}
d\lambda_t = \drift_\lambda(\lambda,t) dt + \vol_\lambda{\lambda_t} dW^2_t,
\label{eq:genericDriftEquation}
\end{equation}
where $W^2_t$ is a Brownian motion independent of $W^1_t$.
We assume that it is possible to fully insure idiosyncratic
longevity risk. Survivors receive longevity credits so that
their wealth process satisfies
\begin{equation}
dw_t = (\lambda_t w_t k + r(w_t-q_t S_t)-c_t+q\mu S_t )dt + q_t \sigma S_t dW^1_t,\label{eq:wealth_eq}
\end{equation}
where $q_t$ is the quantity invested in the risky asset at time $t$. $k=0,1$ controls whether the longevity credit system is implemented or not.

Our problem has two symmetries. Firstly, at any time $t$, the Black-Scholes market with initial stock price $S_t$, is equivalent to the same market with any other possible initial stock price. As a result, one expects that the optimal investment strategy will be independent of the stock price $S_t$. Secondly, we recall that the discrete-time
Epstein-Zin utility $Z$, is positively homogeneous and that $V$
is derived as being the discrete-time limit of $\alpha V^\alpha$. Hence, $V$ is positively homogeneous of order $\alpha$, that is
$V(\zeta c,\lambda)=\zeta^\alpha V(c, \lambda)$. Since we are in the Black-Scholes market, one expects that the value function will also be positively homogeneous of order $\alpha$ in the wealth.
This motivates an ansatz for the HJB equation.

\begin{proposition}
If one substitutes the ansatz
\begin{equation}
V(t,\lambda,w, S)=w^\alpha g(\lambda,t)
\label{eq:ansatz}
\end{equation}
into the HJB equation for the problem of optimizing Epstein--Zin preferences with mortality in an insured drawdown fund,
with dynamics given by equations \eqref{eq:wealth_eq} and \eqref{eq:genericDriftEquation}, one obtains the
partial differential equation:
\begin{align}
    &\alpha g\Bigg(  -\frac{\mu ^2}{2(\alpha -1) \sigma ^2} + (\alpha  g)^{\frac{\rho }{(\rho-1)\alpha }}\left(\frac{1}{\rho }-1\right)+\nonumber\\
   & \lambda  \left(k -\frac{1}{\alpha}\right)-\frac{r^2}{2(\alpha -1) \sigma ^2}+ r
   \left(\frac{\mu }{\left(\alpha -1\right) \sigma ^2}+1\right)\Bigg)+\nonumber\\
    &\frac{\partial g}{\partial t} + \drift_\lambda \frac{\partial g}{\partial \lambda}+\frac{1}{2} \vol_\lambda \frac{\partial^2 g}{\partial \lambda^2}=0\label{eq:g_eqtn}.
\end{align}
The optimal consumption rate and quantity invested in the stock
are given by:
\begin{align}
    & c^*= w (\alpha g) ^{\frac{\rho}{\alpha(\rho-1)}} \label{eq:op_consumption2},\\
    & q^*=\frac{w(\mu-r)}{S(1-\alpha) \sigma^2 }\label{eq:op_quantity2}.
\end{align}
\end{proposition}
Note that the equation for the optimal investment is identical to the optimal investment strategy in the classical Merton problem.

\subsection{A stylised mortality model}\label{sec:stylised}

We now consider a concrete mortality model
that is chosen to ensure the problem has analytic solutions.
To do this, we assume that the mortality satisfies the SDE
\begin{equation}
    \ed \lambda= a \lambda^2 \, \ed t + b \lambda^{\frac{3}{2}} \, \ed W_t.\label{eq:stylised_mortality_SDE}
\end{equation}
for some constants $a$ and $b$.

This model is not intended to be a realistic model of human mortality, indeed, it significantly exaggerates systematic longevity 
risk. However, this model does have similar qualitative properties to human mortality: $\lambda_t$ will always be positive, it will explode to $+\infty$ in a finite time, ensuring that all members of the funds die in finite time and,
ignoring short-term fluctuations, mortality rates increase with age. 

The key motivation for using this mortality model is that
it is designed to lead to an analytically solvable model. This
is because the equation is time-scale invariant: if one changes the units of time, then equation \eqref{eq:stylised_mortality_SDE} is unchanged. If we also assume that $\mu=r=0$ in our market model and that $\delta=0$ in our preferences, this eliminates all time-scale dependence in our equations. Thus, with these assumptions, the HJB equation will have a scaling symmetry in time. In addition, since all our coefficients are constant, the HJB equation has a translation symmetry in time.
This suggests we use the ansatz
\begin{equation}
    g(\lambda) = B \lambda^\xi,\label{eq:style_ansatz}
\end{equation}
for equation \eqref{eq:g_eqtn}.
It is now a simple calculation to check that
\begin{theorem}
The HJB equation for insured drawdown investment ($k=1$)
with mortality model \eqref{eq:stylised_mortality_SDE} and with $\mu=r=\delta=0$, has the trivial solution zero together
with an additional analytic
solution
\begin{equation}
    V=w^\alpha \frac{\lambda ^{\frac{\alpha  (\rho -1)}{\rho }} \left(\frac{(k\alpha -1) \rho }{\alpha  (\rho -1)}+a+\frac{b^2 (\alpha  (\rho -1)-\rho )}{2 \rho
   }\right)^{\frac{\alpha  (\rho -1)}{\rho }}}{\alpha }\label{eq:analytical},
\end{equation}
so long as the term in brackets is positive.
\label{thm:stylisedMortality}
\end{theorem}

Rigorously proving that
our analytic solution is the unique solution to the optimization problem would be somewhat challenging, but we believe the papers \cite{optimal_EZ_finite}, \cite{optimal_EZ_infinite} provide a basic roadmap for how one might do this. We believe the key step in such a proof will be to replace the time variable $t$, with a new variable $s$, satisfying the equation:
\[
ds = \lambda_t dt.
\]
In this new coordinate system, the problem becomes an infinite-horizon Epstein--Zin utility control problem with $\lambda_s$ now following a geometric Brownian motion, rather than exploding in finite time.

When the optimization problem is ill-posed, the supremum of the value function may be $-\infty$, $0$ or $\infty$. This explains why
we sometimes obtain complex solutions to the HJB equation. By evaluating the value-function for other strategies that fit the ansatz, but which do give real values, one can analyze the supremum of the value function in those cases where the supremum is not attained. See Appendix \ref{sec:complexSolutions} for an illustrative example.

Given an optimal drawdown problem, we can set the value of $b$ to zero to obtain an associated drawdown problem without any systematic longevity risk. In order for the original drawdown problem to have the same value function as the drawdown problem without systematic longevity risk, we will need to change the initial wealth by a certain percentage relative to the associated problem without systematic risk. 
We will call this percentage ``the cost of systematic mortality risk'' for the problem. We plot the cost of systematic mortality risk for different values of $\alpha$ and $\rho$ in Figure \ref{fig:costOfSystematicMortalityRisk}, in the case $a=4$, $b=1$ and $\lambda(0)=0.01$. The plot is complex,
but as a brief summary one can see that the cost of systematic mortality risk varies within each quadrant in either a clockwise or anti-clockwise direction, jumping as one moves from quadrant to quadrant. Note that only the qualitative features of this figure
are important as we have made no attempt to calibrate the model
to data.

The most interesting point is that the cost of systematic mortality risk may be either positive or negative depending on the values of $\alpha$ and $\rho$. This may seem surprising at first, but it can be explained by the tension between the change in risk that occurs as $b$ changes with the increase in information that also occurs (this tension is balanced along the curve $\alpha=\rho/(\rho-1)$). This highlights an important difference between an annuity and insured drawdown. An annuity provider sets a price at retirement, but in insured drawdown it is possible to react to new information as it comes in and adjust one's consumption strategy.
The supremum of the value function for the insured drawdown problem with systematic longevity risk will always be greater than the value function for the insured drawdown problem without systematic longevity risk, if the mortality distributions conditioned on the information at time zero are the same. This is because receiving information increases the size of the set of admissible strategies, and so can only increase the value. By contrast the cost of an annuity will be expected to go up in this situation, making the strategy of purchasing an annuity at retirement less attractive.

To give a more precise statement, we specialise to the case of von Neumann-Morgenstern utility, where $\alpha=\rho$. The value function then simplifies to
\begin{equation}
V = \frac{w^\alpha}{\alpha} \lambda^{\alpha-1}
\left( \frac{1-k}{1-\alpha} + a + \frac{b^2 (\alpha-2)}{2} \right)^{\alpha-1}.
\end{equation}
Hence, when $\alpha>0$ we find that the value is decreasing in $a$ and increasing in $b$. Given our remarks earlier on the relationship between the sign of $\alpha$ and pension adequacy, we see that this corresponds to an investor with an adequate pension who will be glad if $a$ is small as they will have more time to enjoy their pension, and will also benefit from knowing if their mortality has increased so they can increase spending. When $\alpha<0$, the value is increasing in $a$ and decreasing in $b$. In this case, the investor has an inadequate pension and so they are concerned if $a$ is too low that they will not have enough to live on. It also appears more difficult for such an investor to benefit from increased information about mortality, and the benefit of information is outweighed by the increased risk.

\begin{figure}
    \centering
    \begin{minipage}{0.6\textwidth}
        \centering
        \includegraphics[width=1.0\textwidth]{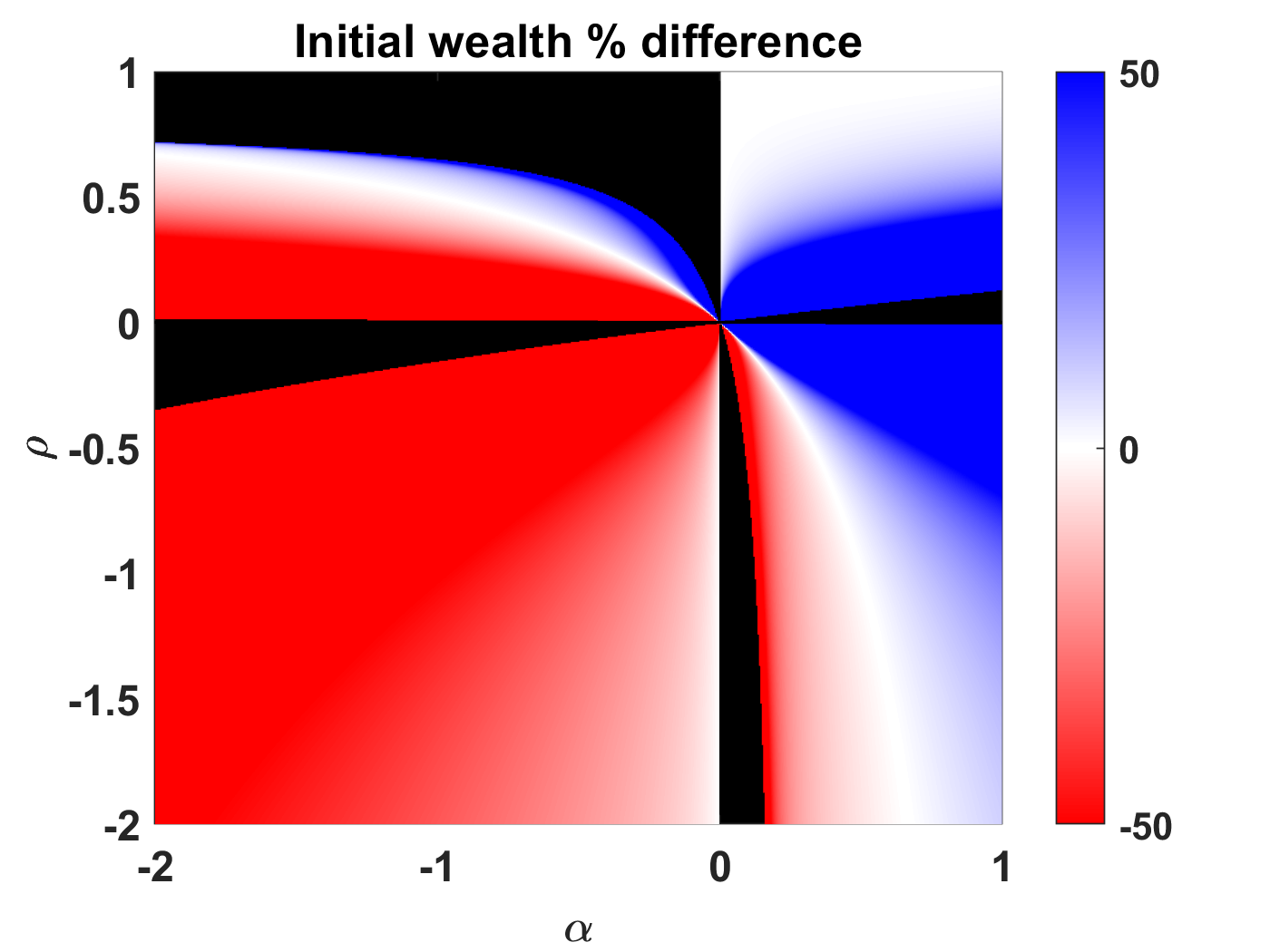}    \end{minipage}
        \caption{The cost of systematic longevity risk
        in our model when $a=4$, $b=1$, $\mu=r=0$ and $\lambda=0.01$.
        Regions in black correspond to points where the analytic solution of the problem, or the associated problem without systematic longevity risk, is complex, so one or other of the problems is ill-posed. 
        The same shade of red is used whenever the cost of longevity risk is 50\% or greater, and the same shade of blue is used whenever the cost of longevity risk is -50\% or less.
        }
    \label{fig:costOfSystematicMortalityRisk}
\end{figure}

\subsection{A realistic mortality model: Cairns-Blake-Dowd}\label{sec:cbd}
We now consider a realistic mortality model to provide meaningful numbers on the differences in performance between a insured drawdown fund with and without systematic longevity risk.
We consider a 1-factor continuous time analogue of the two-factor Cairns-Blake-Dowd model, which leads to the following mortality rate equation:
\begin{multline}
    d\lambda=((e^{\lambda} - 1)(B_1 t(e^{\lambda} - 1) + B_1 t + B_2(B_3 t^2 +\\ B_4 t + B_5) + B_6 e^{\lambda})/((e^{\lambda} - 1)^2 + 2 e^{\lambda} - 1)) dt \\
    + (B_7 (e^{\lambda} - 1)(B_8 t^2 + B_9 t + 1)^{\frac{1}{2}} e^{-\lambda}) dW_t. \label{eq:CBD_lambda}  
\end{multline}
The constant coefficients $B_i$ can be found in \Cref{table:coefficients}. Details of the derivation of this equation can be found in \Cref{sec:CBD_details}.
\begin{table}
\begin{centering}
\begin{tabular}{ll}
 \hline
     Coefficient & Value \\
    \hline
     $B_1$ &  0.00118 \\
     $B_2$ &  0.00306 \\
    $B_3$ &  $1.08\times 10^{-5}$ \\
    $B_4$ &  0.00140 \\
    $B_5$ &  1.05 \\
    $B_6$ &  0.137 \\
    $B_7$ &  0.0799 \\
    $B_8$ & $1.03\times 10^{-5}$ \\
    $B_9$ &  0.00134 \\
    \hline
\end{tabular}
\caption{Coefficients for the SDE \eqref{eq:CBD_lambda}}
\label{table:coefficients}
\end{centering}
\end{table}
\noindent
A fan diagram of simulated mortality rates and the probability density function of the time of death, is shown in \Cref{fig:CBD}.
\begin{figure}
    \centering
    \begin{minipage}{0.47\textwidth}
    \centering
        \includegraphics[width=1.0\textwidth]{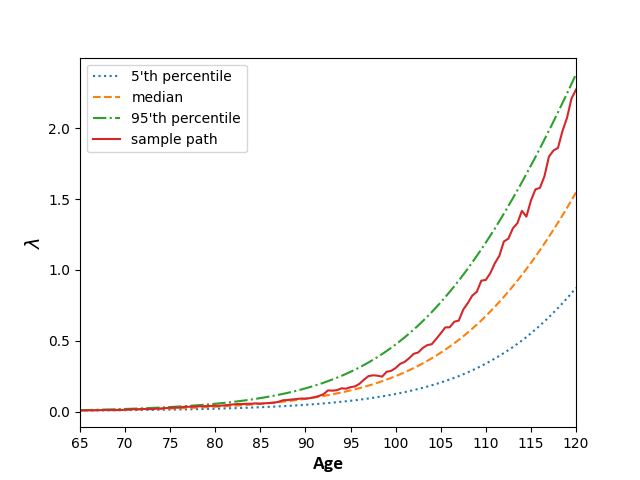}\\
        \ \textrm{(a)}
    \end{minipage}
    \begin{minipage}{0.47\textwidth}
    \centering
        \includegraphics[width=1.0\textwidth]{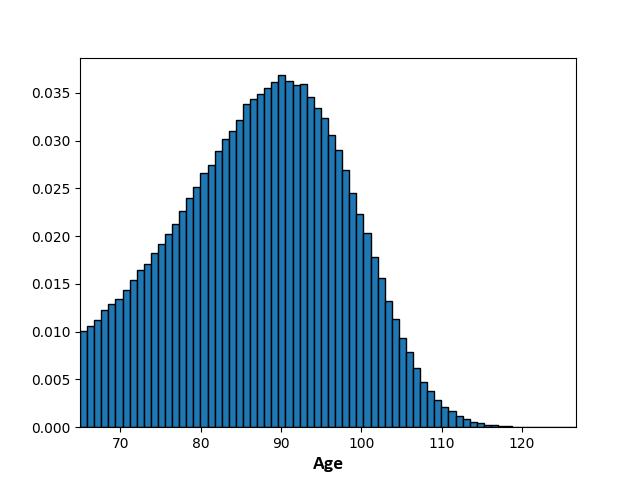}\\
        \ \textrm{(b)}
    \end{minipage}
        \caption{(a) Fan diagram of the mortality rate \eqref{eq:CBD_lambda} and (b), the probability density function of the time of death. Both are generated from $10^6$ scenarios with 100 equally spaced time steps.}
    \label{fig:CBD}
\end{figure}

For this mortality model, our HJB equation is \eqref{eq:g_eqtn} with $\drift_\lambda$ and $\vol_\lambda$ replaced by the drift and volatility in \eqref{eq:CBD_lambda}. Full details of our numerical approach for solving this equation can be found in \Cref{sec:numerical_approach}. Our choice of market parameters are as in \Cref{sec:discrete_model}. Regarding choices of $\alpha$ and $\rho$, recall from \Cref{sec:complexSolutions}, that when $\alpha,\rho$ are of opposite sign in our stylised model, a maximum either does not exist, or the value function is zero. For the CBD model, numerical results suggest we obtain a value function of zero in these cases. Hence, we focus on cases when $\alpha$ and $\rho$ are of the same sign and meaningfully solutions should exist.

In \Cref{table:longevity table}, we provide a summary of the cost of systematic longevity risk for different values of $\alpha$ and $\rho$. The table follows the same qualitative trends predicted by our stylised mortality model. That is, for $\alpha>0$ $(\rho>0)$ the cost is negative and for $\alpha<0$ $(\rho<0)$ the cost is positive (when $\alpha$ and $\rho$ are of opposite sign, we find the value function to be zero when systematic longevity risk is and is not present, at least for the parameter values we have checked). 
Focusing on the case $\alpha<0$, where systematic longevity risk has a detrimental impact on the fund, we see provided $\alpha$ is not too small (i.e., $-2\leq\alpha<0$),  systematic longevity risk has around a 6\% (or less) impact on the starting wealth a fund requires. As a result, systematic longevity may not be considered a significant risk factor for most collective funds that are not highly risk averse. However, for more extreme choices of $\alpha,\rho$, e.g. $\alpha=-10,\rho=-1$, the cost of systematic longevity risk can be as much as 22.4\%. In this situations there is likely to be considerable benefit from purchasing insurance contracts on systematic longevity risk, a question we will address in future work. 

\begin{table}
\begin{tabularx}{\textwidth} { 
  | >{\raggedright\arraybackslash}X 
  | >{\centering\arraybackslash}X 
  | >{\centering\arraybackslash}X
  | >{\centering\arraybackslash}X
  | >{\centering\arraybackslash}X
  | >{\centering\arraybackslash}X
  | >{\raggedleft\arraybackslash}X | }
 \hline
   & $\alpha=-10,\newline \rho=-1$& $\alpha=-2,\newline \rho=-1$ & $\alpha=-1/2,\newline \rho=-1$ & $\alpha=-1/2,\newline \rho=-1/2$ & $\alpha=1/2, \newline\rho=1/2$& $\alpha=3/4,\newline \rho=3/4$\\
 \hline
 Initial wealth  difference & 22.4\% & 6.45\% & 3.06\% & 6.02\% & -7.34\%
 & -22.8\%
 \\
 \hline
\end{tabularx}
\caption{Cost of systematic longevity risk under the CBD mortality model.}
\label{table:longevity table}
\end{table}

With the impact of systematic longevity quantified, we now analyse how a collective fund in the presence of systematic longevity risk performs in comparison to an annuity. As in \Cref{sec:discrete_model}, we use an initial fund value of £126,636 for our collective scheme, which under our CBD mortality model purchases an annuity of approximately £5,100 a year. In the remainder of this section, we consider von Neumann-Morgenstern preferences with $\alpha=\rho$ for ease of computation.

Fan diagrams of the consumption streams with and without investment in a risky asset when $\alpha=\rho=0.5$, are presented in \Cref{fig:annuity_comparison}. With investment in a risky asset, the median scenario for the collective fund outperforms an annuity with a total consumption of £900,000 (which is unrealistically high). The median scenario delivers a highly inconsistent consumption stream with little consumption in the early years of retirement and negligible consumption from age 115 onwards, but with extremely high levels of consumption in between this range. When there is no risky asset, the median consumption is above the annuity level until the age of 90 and falls to zero by age 110. The behaviour here is analogous to that in \Cref{fig:epstein-zin-with-riskaversion} (a) and for the same reasons we do not believe these plots represent a realistic consumption strategy.

We repeat the same exercise for $\alpha=\rho=-0.5$ in \Cref{fig:annuity_comparison_negative_alpha}. With investment in a risky asset the collective fund outperforms an annuity, the median scenario is always above the annuity, while the 5'th percentile of scenarios dips slightly below the annuity in some years. When no risky asset is present, the median scenario outperforms the annuity up to age 90 and from age 110 onwards. These consumption strategies seem to be more realistic than the those seen in the previous figure and mirror those in \Cref{fig:epstein-zin-with-riskaversion} (b). Both \Cref{fig:annuity_comparison} and \Cref{fig:annuity_comparison_negative_alpha} highlight the value of longevity pooling and investing in risky assets post retirement.

\begin{figure}
    \centering
    \begin{minipage}{0.49\textwidth}
        \includegraphics[width=1.0\textwidth]{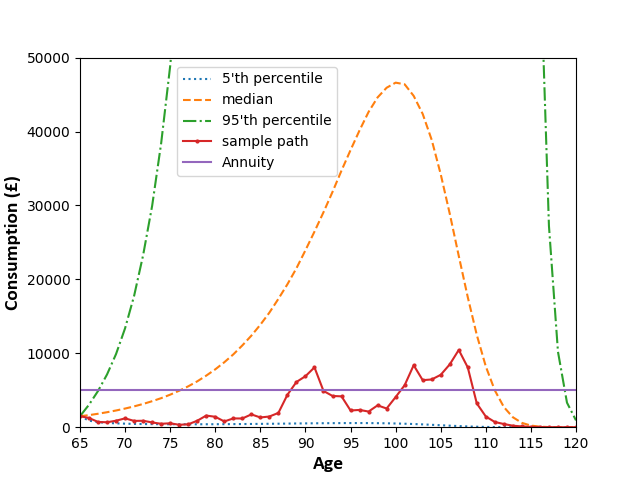}
    \end{minipage}
    \begin{minipage}{0.49\textwidth}
        \centering
        \includegraphics[width=1.0\textwidth]{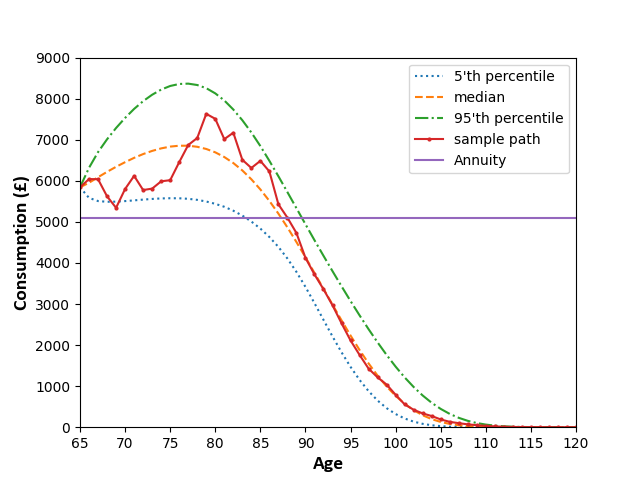}\\
    \end{minipage}
        \caption{Consumption stream with (left) and without (right) a risky asset for a pooled fund with systematic longevity risk. The funds start with an initial wealth of £126,636. Here, $\alpha=\rho=0.5$, $10^6$ scenarios are generated and 90 equally spaced time steps are used.}
    \label{fig:annuity_comparison}
\end{figure}

\begin{figure}
    \centering
    \begin{minipage}{0.49\textwidth}
        \includegraphics[width=1.0\textwidth]{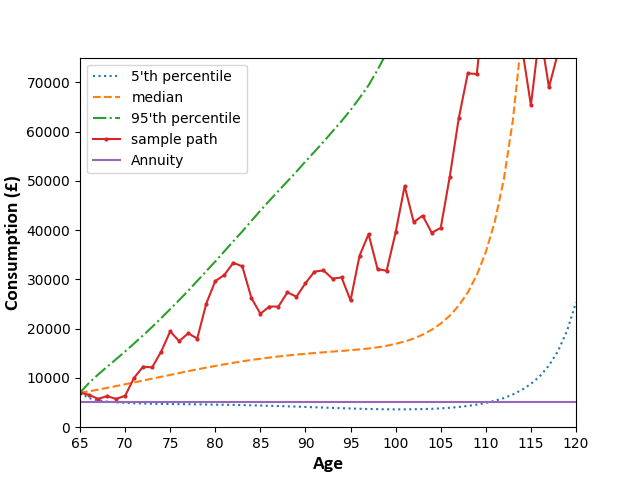}
    \end{minipage}
    \begin{minipage}{0.49\textwidth}
        \centering
        \includegraphics[width=1.0\textwidth]{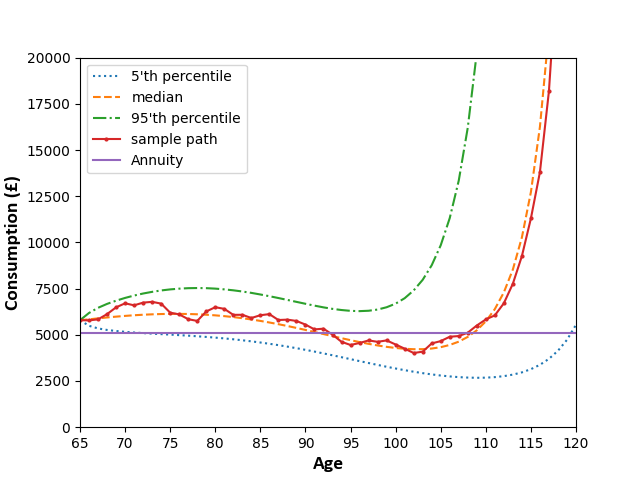}\\
    \end{minipage}
        \caption{Consumption stream with (left) and without (right) a risky asset for a pooled fund with systematic longevity risk. The funds start with an initial wealth of £126,636. Here, $\alpha=\rho=-0.5$, $10^6$ scenarios are generated and 90 equally sized time steps are used.}
    \label{fig:annuity_comparison_negative_alpha}
\end{figure}

\section{Conclusions}\label{sec:conclusions}

In this paper, we study optimal consumption-investment problems for collective pension funds characterised by a tontine/longevity credit system within the Black-Scholes market. In the interest of analytic tractability, in both discrete and continuous time, we model investors preferences over pension outcomes through Epstein-Zin preferences. This framework allows us to provide informative analytic formulae which illuminate the impact of investors preferences and longevity risks on pension outcomes.

In discrete time, we consider the case where investors are exposed to idiosyncratic longevity risk only. For a homogeneous collective, we solve the problem via difference equations that determine  the optimal strategy for both finite and infinite funds. This allows us to elucidate the effects of risk aversion, $\alpha$, and satiation, $\rho$, on consumption with and without market risk. 
Furthermore, applying the optimal strategy for an infinite homogeneous collective to each investor in a heterogeneous collective, we see the resulting outcomes are still effectively optimal. Thus, the solution to the homogeneous case provides a useful tool for the management of complex heterogeneous funds.

In continuous time, we additionally incorporate systematic longevity risk. Using a stylised mortality model, we analytically solve the associated HJB equation and explicitly obtain the optimal strategies. This reveals the surprising fact that systematic longevity risk can improve pension outcomes provided $\alpha>\rho/(\rho-1)$. This occurs because the benefits of gaining additional information about one's mortality over time outweigh the detrimental effect of changing mortality risk. These effects are balanced when $\alpha=\rho/(\rho-1)$, while risk outweighs information for $\alpha<\rho/(\rho-1)$, leading to worse outcomes. These trends remain true for the case of a more realistic 1-factor continuous time CBD mortality model, at least for $\alpha,\rho>0$ and $\alpha,\rho<0$. Here, we solve the HJB equation numerically and provide estimates on the cost of systemic longevity risk. Provided individuals are not too strongly risk averse ($\alpha<0$ is not too small), systematic longevity risk requires approximately 6\% (or less) additional wealth to be put into the fund initially to achieve equivalent outcomes.

Our results demonstrate that, as would be expected,  there is considerable benefit to be gained by investing in a collective fund with longevity-risk pooling over a fixed-rate annuity. This is because of the ability to include some investment in the risky asset, the ability to benefit from inter-temporal substitution and the ability to adaptively vary one's strategy. In the face of systematic longevity risk, these advantages become more striking: hedging systematic longevity risk is extremely challenging for an annuity provider, whereas a collective fund can gradually vary consumption levels in response to new information about mortality.

An additional conclusion from our work is the need for better preference models. While Epstein-Zin preferences lead to tractable models, and they have been successful in solving asset pricing puzzles, the consumption strategies they produce are not always reasonable. The key issue is that for $\alpha>0$, investors are willing to accept little or no income in some years, while for $\alpha<0$ individuals try to target an infinite income as they approach death. This is a result of the positive homogeneity of Epstein-Zin preferences. Hence, inhomogeneous optimization models that impose minimum consumption or adequacy levels (which every individual will need in retirement), are required.

\section*{Acknowledgments}
JA and JD gratefully acknowledge funding from the Nuffield Foundation (grant FR-000024058).


\bibliographystyle{abbrvnat}
\bibliography{collectivization}

\appendix

\section{Proofs and technical appendices}
\label{appendix:longevitypooling}

\subsection{The optimal control problem with no systematic longevity risk}

From our assumption of non systematic longevity risk, we have independent identically distributed random variables $\tau_i$ representing the last time of consumption of individual $i$. Write $n_t$ for the number of survivors at time $t$, that is the number of individuals whose where $\tau_i>t$.
We will write down optimal control problems for both finite
and infinite collectives.

Write $a_t$ for proportion of the fund invested in stock at time $t$.
Write $X_t$ for the value of the fund per survivor at time
$t$ before consumption.
Define $X_t=0$ if $n_t=0$.
Similarly, let ${\overline{X}}_t$ denote the value of the fund per
survivor after consumption. Note that $X_t=\lim_{h\nearrow t} \overline{X}_h$ at time points $t \in {\cal T}$. At intermediate
times, $t \in [i\delta t, (i+1)\delta t)$, ${\overline{X}}_t$
obeys the SDE
\[
\ed \overline{X}_t = \overline{X}_t( a_t \mu + (1-a_t) r) \, \ed t
+ \overline{X}_t a_t \sigma \, \ed W_t,
\]
with initial condition given by the budget equation
\begin{equation}
\overline{X}_t = \begin{cases}
\frac{n_{t}}{n_{t+\delta t}} (X_t - c_t) & n \text{ finite} \\
s_t^{-1} (X_t - c_t) & n=\infty
\end{cases}
\label{eq:budgetEquationGeneral}
\end{equation}
unless $n_{t+\delta t}=0$ in which case we define $\overline{X}_t$ to be zero on $[t, t+\delta t)$. Observe that this formula
is based on our convention that an individual who dies at a time $t$ still consumes at that time and the corresponding convention for $n_t$ which ensures $n_{t+\delta t}$ is ${\cal F}_t$ measurable.

Let $\tilde{\cal A}(x,t_0)$ denote the space of admissible controls $(c_t,a_t)$: that is ${\cal F}_t$-predictable processes such that
$0 \leq c_t \leq X_t$ and with $X_{t_0}=x$. We define the value
function of our problem starting at time $t_0$ to be
\begin{equation}
v_n(x,t_0) = \sup_{(c,a) \in \tilde{\cal A}(x,t_0)} Z_{t_0}(c, \tau_1).
\label{eq:defValueFunction}
\end{equation}
where $Z_{t_0}$ is an Epstein--Zin utility function.

\subsubsection{The cases \texorpdfstring{$n=1$}{n=1} and \texorpdfstring{$n=\infty$}{n=infinity}}
\label{sec:epsteinZinFormulae}

To avoid notational complexity we consider separately the
simple cases when $n=1$ and $n=\infty$,
leaving the case of general $n$ until Section \ref{sec:epsteinZinFiniteCollective}.
We define $C$ to distinguish these cases as in equation \eqref{eq:defC}.

We write $z_t:=v_n(1,t)$ so
that the positive-homogeneity of $Z_t$ implies that
\begin{equation}
v_n(x,t)=x\, z_t.
\label{eq:zscaling}
\end{equation}
We note that we have not yet shown whether $z$ is finite, but equation \eqref{eq:zscaling} can still be interpreted for infinite values of $z$.

Let $c^r_t$ denote the consumption rate at time $t$, so $c_t=c^r_t X_t$
for an individual who is still alive at time $t$.
The budget equation \eqref{eq:budgetEquationGeneral} 
can then be written as
\begin{equation}
\overline{X}_t = s_t^{-C} (1-c^r_t) X_t.
\label{eq:budgetEquation}
\end{equation}

We now let ${\cal A}_{t+{\delta t}}(x,c^r_t)$ be the set of random variables
$X_{t+\delta t}$, representing the value of our investments at time $t+\delta t$,
that can be obtained by a continuous time trading strategy with
an initial budget given by \eqref{eq:budgetEquation} with $X_t=x$. Thus
${\cal A}_{t+\delta t}(x,c^r_t)$ is the set of admissible investment returns given the budget
and the consumption.

The Markovianity of Epstein--Zin preferences and the definition of $Z_t$ to compute allow us to apply the dynamic programming principle to compute
\begin{align}
v_n(x,t) &= \sup_{c^r_t \geq 0, X_{t+\delta t}\in {\cal A}_{t+\delta t}(x,c^r_t)}
[
(x\, c^r_t)^\rho
+ \beta \{s_t \E_t( v_n(X_{t+\delta t}, t+\delta t)^\alpha )\}^\frac{\rho}{\alpha}
]^\frac{1}{\rho} \nonumber \\
&= \sup_{c^r_t \geq 0, X_{t+\delta t}\in {\cal A}_{t+\delta t}(x,c^r_t)}
[
(x\, c^r_t)^\rho
+ \beta z_{t+\delta t}^\rho  \{ s_t \E_t( X_{t+\delta t}^\alpha )\}^\frac{\rho}{\alpha}
]^\frac{1}{\rho}.
\label{eq:valueFunction1}
\end{align}
The second line follows from the first by the positive homogeneity of Epstein--Zin utility \eqref{eq:zscaling}.

If $\alpha>0$, we may compute the value of 
\[
\sup_{X_t \in {\cal A}_{t+\delta t}(c^r_t)} \E_t( X_{t+\delta t}^\alpha )
\]
by solving the Merton problem for optimal investment over time
period $[t, t+ \delta t]$, with initial budget $\overline{X}_t$, no consumption,
and terminal utility function $u(x)=x^\alpha$.

We find
\begin{equation}
\sup_{X_t \in {\cal A}_{t+\delta t}(c_t)} \E_t( X_{t+\delta t}^\alpha ) = (\exp( \xi \, \delta t) \overline{X}_t)^\alpha
\label{eq:mertonOptimum}
\end{equation}
where
\begin{align}
\xi &= \sup_{a \in \R}[ a(\mu - r) + r - \frac{1}{2}a^2(1-\alpha)\sigma^2]
\end{align}
For details see Merton's paper \cite{merton1969lifetime} or \cite{pham} equations (3.47) and (3.48). We find that $\xi$ is as given in equation \eqref{eq:defxi}
and moreover the proportion invested in stock
is given by $a^*$ which is a constant determined entirely by $\alpha$ and the market.
In the case where $\alpha<0$ we must instead compute
\[
\inf_{X_t \in {\cal A}_{t+\delta t}(c_t)} \E_t( X_{t+\delta t}^\alpha ).
\]
However, apart from the change of a $\sup$ to an $\inf$ everything is algebraically identical, so the same formulae emerge.

Putting the value \eqref{eq:mertonOptimum} into our
expression \eqref{eq:valueFunction1} for the value function we obtain
\begin{align}
v_n(x,t) &= 
\sup_{c^r_t\geq 0}
[
(x\, c^r_t)^\rho
+z_{t+\delta t}^\rho \beta \, (s_t \exp( \alpha \xi \, \delta t) \overline{X}_t^\alpha)^\frac{\rho}{\alpha} ]^\frac{1}{\rho}
\nonumber \\
&= 
\sup_{c^r_t\geq 0}
[
(x\, c^r_t)^\rho
+\beta \, (z_{t+\delta t} \exp( \xi \, \delta t) s_t^{(\frac{1}{\alpha}-C)} (1-c^r_t) X_t)^{\rho} ]^\frac{1}{\rho}
\label{eq:valueFunction2}
\end{align}
where the last line follows from the budget equation \eqref{eq:budgetEquation}.
We define
\begin{equation}
\phi_t:=\beta^\frac{1}{\rho} \exp(\xi \delta t) s_t^{(\frac{1}{\alpha}-C)}
\quad \theta_t:=\phi_t z_{t+\delta t}
\label{eq:defphi}
\end{equation}
so that \eqref{eq:valueFunction2} may be written as
\begin{align}
z_t &= 
\sup_{c^r_t\geq 0}
[
(c^r_t)^\rho
+\theta_t^\rho (1-c^r_t)^\rho ]^\frac{1}{\rho}
\label{eq:valueFunction3}
\end{align}

Differentiating the expression in square brackets on the right-hand
side, we see that the supremum is achieved in equation \eqref{eq:valueFunction2}
when $c^r_t=c^*_t$, where $c^*_t$ satisfies
\[
\rho \, (c^*_t)^{\rho-1} - \rho \theta_t^\rho((1-c^*_t)^{\rho-1} = 0,
\]
or, if this yields a negative value for $c^*_t$, we should take $c^*_t=0$.
Simplifying we must have
\begin{equation}
\left( \frac{c^*_t}{1-c^*_t} \right)^{\rho-1} = \theta_t^\rho.
\label{eq:cStarRelation}
\end{equation}
So
\begin{equation}
c^*_t = 
(1+\theta_t^\frac{\rho}{1-\rho})^{-1}.
\label{eq:defcstar}
\end{equation}
This expression for $c^*_t$ is non-negative, so it always gives the argument for the supremum  in \eqref{eq:valueFunction2}. We obtain
\begin{align}
z_t
&=
[(c^*_t)^\rho
+  \theta_t^\rho (1-c^*_t)^{\rho} ]^\frac{1}{\rho} \nonumber \\
&=
c^*_t \left[ 1
+  \theta_t^\rho \left(\frac{1-c^*_t}{c^*_t}\right)^{\rho} \right]^\frac{1}{\rho} \nonumber \\
&=
c^*_t \left[ 1
+  \theta_t^\rho \left(\theta_t^\frac{\rho}{\rho-1} \right)^{-\rho}
\right]^\frac{1}{\rho}, \quad \text{by \eqref{eq:cStarRelation}}, \nonumber \\
&=
c^*_t \left[ 1
+  (\theta_t^{1-\frac{\rho}{\rho-1}})^{\rho}
\right]^\frac{1}{\rho} \nonumber \\
&=
c^*_t \left[ 1
+  \theta_t^\frac{\rho}{1-\rho}
\right]^\frac{1}{\rho} \nonumber \\
&=
(1 + \theta_t^\frac{\rho}{\rho-1})^{-1} ( 1
+  \theta_t^\frac{\rho}{1-\rho}
)^\frac{1}{\rho}, \quad \text{ by \eqref{eq:defcstar}}, \nonumber \\
&=
( 1 +  \theta_t^\frac{\rho}{1-\rho})^\frac{1-\rho}{\rho}. \nonumber
\end{align}
We conclude that
\begin{equation}
z_t^\frac{\rho}{1-\rho} = 1 + \theta_t^\frac{\rho}{1-\rho} = 1 + \phi_t^\frac{\rho}{1-\rho} z_{t+\delta t}^\frac{\rho}{1-\rho}
\label{eq:valueFunction4}
\end{equation}
where $\phi$ is given by equation \eqref{eq:defphi}. We observe also that equation \eqref{eq:defcstar} for the optimal consumption rate per survivor simplifies to \eqref{eq:cStarConclusion}.

This completes the proof of Theorem \ref{thm:ezIdioscynractic} in the cases $n=1$ and $n=\infty$.

\subsubsection{Finite funds}
\label{sec:epsteinZinFiniteCollective}

Let $v_n$ be the value function \eqref{eq:defValueFunction} for
the optimal investment problem for $n$ individuals with
homogeneous Epstein--Zin preferences. By positive homogeneity
we may define $z_{n,t}:=v_n(1,t)$, so that $v_n(x,t)=x z_{n,t}$.

Let $I_{i,t_0+\delta t}$ denote the event that both
\begin{enumerate}[(i)]
	\item there are $i$ survivors at time $t_0+\delta t$, i.e.\ $n_{t+\delta t}=i$.
	\item the individual whose utility we wish to calculate is one of those survivors, so $\tau^\iota>t$.
\end{enumerate}
Recall that $X_{t}$ denotes the value of the fund per survivor at time $t$ before consumption
and mortality.
\[
\E_{t_0}(Z_{t_0+\delta t}^\alpha \mid I_{i,{t_0+\delta t}}) = 
\E_{t_0}( v_i( X_{t_0+\delta t}, t_0 + \delta t)^\alpha \mid I_{i,t_0+\delta t}).
\]
Hence
\[
\E_{t_0}(Z_{t_0+\delta t}^\alpha) = \sum_{i=1}^{n_t} \frac{i}{n_t} S_t(n_t,i) \E_{t_0}(v_i( X_{t_0+\delta t}, t_0 + \delta t)^\alpha \mid I_{i,t_0+\delta t}).
\]
We now let ${\cal A}_{t+{\delta t}}(x,c^r_t)$ be the set of random variables
$X_{t+\delta t}$ representing the value of the fund per survivor at time $t+\delta t$
before consumption that can be obtained by a continuous time trading strategy given initial capital $\overline{X}_t = \frac{n_{t+\delta t}}{n_t} (1-c^r_t) X_t$ per survivor when $X_t=x$. Using the dynamic programming principle we find
\begin{small}
\begin{align}
&v_n(x,t) \nonumber \\
&= \sup_{\stackrel{c^r_t \geq 0}{X_{t+\delta t}\in {\cal A}_{t+\delta t}(x,c^r_t)}}
\left[
(x\, c^r_t)^\rho
+ \beta \left\{\sum_{i=1}^n  \frac{i}{n} S_t(n,i) \E_t( v_i(X_{t+\delta t}, t+\delta t)^\alpha \mid I_{i,t_0+\delta t} ) \right\}^\frac{\rho}{\alpha}
\right]^\frac{1}{\rho} \nonumber \\
&= \sup_{\stackrel{c^r_t \geq 0}{X_{t+\delta t}\in {\cal A}_{t+\delta t}(x,c^r_t)}}
\left[
(x\, c^r_t)^\rho
+ \beta \left\{ \left(\sum_{i=1}^n \frac{i}{n} S_t(n,i) z_{i,{t+\delta t}}^\alpha \E_t(  X_{t+\delta t}^\alpha \mid I_{i,t_0+\delta t} ) \right)  \right\}^\frac{\rho}{\alpha}
\right]^\frac{1}{\rho}. \nonumber \\
\end{align}
\end{small}
We used positive homogeneity to obtain the last line. The argument of Section \ref{sec:epsteinZinFormulae} tells us how to optimize over $X_t$.
Hence we find
\begin{equation*}
z_{n,t} =
\sup_{c^r_t \geq 0}
\left[
(c^r_t)^\rho
+ \beta \left(\sum_{i=1}^n \left( \frac{i}{n} \right)^{1-\alpha} S_t(n,i) z_{i,t+\delta t}^{\alpha} \right)^\frac{\rho}{\alpha} \left(\exp( \xi \, \delta t) (1-c^r_t)\right)^\rho  
\right]^\frac{1}{\rho}
\end{equation*}
where $\xi$ is as defined in equation \eqref{eq:defxi}. The optimal investment policy is also described in equation \eqref{eq:defastar}, and as before it depends only on the market and the monetary-risk-aversion parameter $\alpha$.

We may rewrite our expression for $z_{n,t}$ as follows:
\begin{equation}
z_{n,t} =
\sup_{c^r_t \geq 0}
\left[
(c^r_t)^\rho
+ \tilde{\theta}_{n,t}^\rho (1-c^r_t)^\rho  
\right]^\frac{1}{\rho}
\label{eq:valueFunction3Tilde}
\end{equation}
where $\tilde{\theta}$ is as given in equation \eqref{eq:deftildetheta}.

Equation \eqref{eq:valueFunction3Tilde} is structurally identical to equation \eqref{eq:valueFunction3}. Hence from equation
\eqref{eq:valueFunction4} we may deduce that equation \eqref{eq:valueFunction4Tilde}
holds as claimed.

\subsection{Distribution of consumption over time}

\begin{proof}[Proof of Theorem \ref{thm:epsteinZinConsumption}]
	We suppose as induction hypothesis that the
	distribution of $X_t$ is as described at time $t$.
	
	The budget equation \eqref{eq:budgetEquation}
	then tells us that the wealth per survivor after consumption, ${\overline{X}}_t$, satisfies
	\[
	\overline{X}_t = s^{-C}_t (1- z_t^{\frac{\rho}{\rho-1}}) X_t.
	\]
	Hence
	\[
	\log \overline{X}_t \sim
	N( \mu^X_t + \log(s^{-C}_t (1- z_t^{\frac{\rho}{\rho-1}})),
	(\sigma^X_t)^2 ).
	\]
	Our investment strategy from $t$ to $(t+\delta t)$ is a continuous time trading strategy where we hold a fixed proportion of our wealth in stocks at all time. So, in the interval $(t,t+\delta t]$, $X_t$ satisfies the SDE
	\[
	\ed X_t = (1-a^*) r X_t \, \ed t + a^* X_t (\mu \, \ed t + \sigma \, \ed W_t), \qquad X_t=\overline{X}_t.
	\]
	By It\^o's lemma we find
	\begin{align}
	\ed (\log X)_t &= (1-a^*) r \, \ed t + a^* ((\mu-\tfrac{1}{2}a^* \sigma^2) \, \ed t + \sigma \, \ed W_t), \qquad \log X_t = \log \overline{X}_t \nonumber \\
	&= \tilde{\xi} \ed t  + a^* \sigma \, \ed W_t.
	\end{align}
	We deduce that $(\log X_{t+\delta t}-\log \overline{X}_t)$ 
	conditioned on the value of $X_t$ will follow a normal distribution with
	mean $\tilde{\xi} \, \delta t$ and standard deviation $a^* \sigma \sqrt{\delta t}$.
	Moreover the random variable $(\log X_{t+\delta t}-\log \overline{X}_t)$
	is independent of $\overline{X}_t$. 
	
	The sum of independent normally distributed random increments yields a new
	normally distributed random variable, and one can compute the mean and variance by adding the mean and variance of the increments. Hence
	\[
	\log X_{t+\delta t}
	\sim N( \mu^X_{t+\delta t}, (\sigma^X_{t+\delta t})^2 )
	\]
	where
	\begin{equation}
	\mu^X_{t+\delta t} = 
	{\mu}^X_t
	+ \log(s^{-C}_t (1- z_t^{\frac{\rho}{\rho-1}}))
	+ \tilde{\xi} \, \delta t
	\label{eq:muXRecursion}
	\end{equation}
	and
	\begin{equation}
	(\sigma^X_{t+\delta t})^2 = (\sigma^X_t)^2 + (a^*)^2 \sigma^2 \delta t  .
	\label{eq:sigmaXRecursion}
	\end{equation}
	Solving the recursion \eqref{eq:sigmaXRecursion}
	yields equation \eqref{eq:sdLogWealth}.
	The result for $X_t$ now follows by induction.
	
	Equation \eqref{eq:distConsumption} follows from equation
	\eqref{eq:cStarConclusion}. Using \eqref{eq:distConsumption}, \eqref{eq:meanLogWealth}
	we calculate
	\begin{align*}
	\E(& \log c_{t+\delta t} \mid c_t) - \log( c_t )  \\
	&=
	\frac{\rho}{\rho-1} \left( \log(z_{t+\delta t})
	- \log(z_t) \right)
	+ \log( s_t^{-C}) + \log( 1 - z_t^\frac{\rho}{\rho-1}) + \tilde{\xi} \delta t
	\\
	&=
	\frac{\rho}{1-\rho}\left(
	\log(z_t) -\log(z_{t+\delta t})  \right)
	+ \log( s_t^{-C}) + \log\left( \frac{z_t^\frac{\rho}{1-\rho} - 1}{z_t^\frac{\rho}{1-\rho}} \right)
	+ \tilde{\xi} \delta t \\
	&= \log( s_t^{-C} ) +
	\log\left(  \phi_t^{\frac{\rho}{1-\rho}}\right) + \tilde{\xi} \delta t
	\qquad \text{ by equation \eqref{eq:valueFunction4}.}
	\end{align*}
	This completes the proof.
\end{proof}

\subsection{Proof of Theorem \ref{thm:cvgcEpsteinZin}}
\label{appendix:longevitypoolingconvergence}

Proving Theorem \ref{thm:cvgcEpsteinZin} requires
some preliminary material.
Our proof strategy will be to approximate an expectation
involving the binomial distribution with a Gaussian integral which
we can then estimate using Laplace's method.
To get a precise convergence result, we need some estimates
on the rate of convergence
in the Central Limit Theorem. The estimates given in \cite{bhattacharya}
suit our purposes well. For the reader's convenience we will summarize the result we will need.

We begin with some definitions.
A random variable $X$ is said to satisfy Cram\'er's condition if its characteristic
function $f_X$ satisfies
\begin{equation}
\sup \{ |f_X(t)| \mid t> \eta \} < 1
\label{eq:cramer}
\end{equation}
for all positive $\eta$.
Let $\Phi$ be the standard normal
distribution. Given a set $A \subseteq \R$, and a function $g$, $\omega_g(A)$ is defined to equal
\[
\omega_g(A):=\sup\{|g(x)-g(y)| \mid x, y \in A \}.
\]
The set $B_{\epsilon}(x)$ is the ball of radius $\epsilon$ around $x$.

Let $Q_n$ be the appropriately normalized $n$-th partial sum
of a sequence of independent identically
distributed random variables $X^{(i)}$ for which Cram\'er's condition
holds and which have finite moments have
all orders. Appropriately normalized
means normalized such that the central
limit theorem implies $Q_n$ converges
to $\Phi$ in distribution.
Then there exists a constant $c$ such that for any bounded measurable function $g$
\begin{equation}
|\int_\R g \, \ed(Q_n - \Phi)| \leq c\, \omega_g(\R) n^{-\frac{1}{2}}
+ \int \omega_g(B_{c n^{-\frac{1}{2}}} (x)) \ed \Phi(x).
\label{eq:bhattacharya}
\end{equation}

The full result given in \cite{bhattacharya}
is more general and more precise than we need.
Let us explain how the statements are related.
We have simplified our statement to the
one-dimensional case, we have assumed the $X^{(i)}$
are identically distributed and we have assumed all 
moments of $X^{(1)}$ exist. The statement
in \cite{bhattacharya} is therefore more
complex, and in particular involves
additional terms called $\rho_{s,n}$
defined in the one-dimensional case by
\[
\rho_{s,n}:=\frac{1}{n}(\sum_{i=1}^n E|\sigma_n X^{(i)}|^s)
\]
where
\[
\sigma_n^2 := n \left( \sum_{i=1}^n \Var( X^{(i)} ) \right)^{-1}.
\]
Our assumptions on $X$ guarantee that $\rho_{s,n}$ is independent of $n$ and so
we have been able to absorb these terms into
the constant $c$. In addition, our statement
uses Theorem 1 of \cite{bhattacharya},
together with remarks at the end of the second paragraph on page 242 about
Cram\'er's condition.

\bigskip

We are now ready to prove the desired convergence result.

\begin{proof}[Proof of Theorem \ref{thm:cvgcEpsteinZin}]
	We proceed by a backward induction on $t$. The result is trivial for the case $t=T$. We now assume the induction hypothesis
	\[
	z_{n,t+\delta t} = z_{\infty,t+\delta t} + O(n^{-\frac{1}{2}}).	
	\]
	We wish to compute $\tilde{\theta}_{n,t}$, but only the sum
	in the expression \eqref{eq:deftildetheta} is difficult to compute. We
	will call this sum $\lambda_{n,t}$, so
	\begin{equation}
	\lambda_{n,t}:=\sum_{i=1}^n \left( \frac{i}{n} \right)^{1-\alpha} S_t(n,i) z^{\alpha}_{i, t+\delta t}.
	\label{eq:defnlambdant}
	\end{equation}
	
	Heuristically, one can approximate with a Gaussian integral using the Central Limit Theorem and then apply Laplace's method to compute the limit as $n \to \infty$. This motivates the idea of decomposing the sum above into a ``left tail'' for small values of $i$, a central term for values of $i$ near the mean of the Binomial distribution $n p$, and a ``right tail'' for larger values of $i$. We will in fact bound the tails separately (Steps 1 and 2, below),
	and then we will be able to rigorously apply a Central Limit Theorem
	argument to the central term (Step 3).
	
	We compute
	\begin{equation}
	\frac{S_t(n,i-1)}{S_t(n,i)} = \frac{i s_t}{(1 - i + n) (1 - s_t)}.
	\label{eq:exponentialDecay}
	\end{equation}
	We note that
	\[
	i \leq \frac{(n+1) (1-s_t)}{(\lambda -1) s_t+1} \implies \frac{i s_t}{(1 - i + n) (1 - s_t)} \leq \frac{1}{\lambda} \implies \frac{S_t(i-1,n)}{S_t(i,n)} \leq \frac{1}{\lambda}.
	\]
	So we define an integer $N_{\lambda, t, n}$ by
	\[
	N_{\lambda, t, n} := \left\lfloor \frac{(n+1) (1-s_t)}{(\lambda -1) s_t+1} \right\rfloor
	\]
	and then equation \eqref{eq:exponentialDecay} will ensure that we
	have exponential decay of $S_t(i,n)$ as $i$ decreases
	\begin{equation}
	i\leq N_{\lambda, t, n} \implies S(n,i) \leq \lambda^{i-N_{\lambda,t, n}}
	S(n,N_{\lambda, t, n}).
	\label{eq:exponentialdecay}
	\end{equation}
	{\bf Step 1.} We can now estimate the left tail of \eqref{eq:defnlambdant}. There
	exists a constants $C_{1,t}$, $C_{2,t}$ such that
	\begin{align}
	\sum_{i=1}^{N_{3,t,n}}  \left( \frac{i}{n} \right)^{1-\alpha} S_t(n,i)
	&\leq C_{1,t} \int_1^{N_{3,t,n}} 3^{-N_{3,t,n}+i} S(n,N_{3, t, n}) \left(\left(\frac{1}{n}\right)^{1-\alpha}+1 \right)
	\, \ed i \nonumber \\
	&\leq C_{2,t} S(n,N_{3, t, n}) \left(\left(\frac{1}{n}\right)^{1-\alpha}+1 \right).
	\label{eq:lefttail1}
	\end{align}
	To estimate this term, we observe that
	\begin{align*}
	N_{2,t,n}-N_{3,t,n} &= \left \lfloor 
	\frac{(n+1) (1-s_t)}{s_t+1}
	\right \rfloor
	-
	\left \lfloor \frac{(n+1) (1-s_t)}{2 s_t+1}
	\right \rfloor \\
	&\geq \left \lfloor 
	\frac{(n+1) (1-s_t)}{s_t+1}
	-
	\frac{(n+1) (1-s_t)}{2 s_t+1}
	\right \rfloor -2 \\
	&= \left \lfloor 
	\frac{(n+1) s_t (1-s_t)}{(s_t+1)(2 s_t + 1)}
	\right \rfloor -2
	\end{align*}
	Hence by equations \eqref{eq:exponentialDecay} and \eqref{eq:lefttail1} we
	find
	\begin{align}
	\sum_{i=1}^{N_{3,t,n}}  \left( \frac{i}{n} \right)^{1-\alpha} S_t(n,i)
	&\leq C_{2,t} S(N_{2, t, n},n) 2^{-\left \lfloor 
		\frac{(n+1) s_t (1-s_t)}{(s_t+1)(2 s_t + 1)}
		\right \rfloor +2} \left(\left(\frac{1}{n}\right)^{1-\alpha}+1 \right) \nonumber \\
	&\leq C_{2,t} 2^{-\left \lfloor 
		\frac{(n+1) s_t (1-s_t)}{(s_t+1)(2 s_t + 1)}
		\right \rfloor +2} \left(\left(\frac{1}{n}\right)^{1-\alpha}+1 \right). \nonumber
	\end{align}
	which decays exponentially as $n\to \infty$. Our induction hypothesis
	ensures that the $z_{i,t+\delta t}^\alpha$ are bounded, so we may safely
	conclude that
	\begin{equation}
	\lambda_{n,t}:=\left( \sum_{i=N_{3,t,n}}^n \left( \frac{i}{n} \right)^{1-\alpha} S_t(n,i) z^{\alpha}_{i, t+\delta t} \right)
	+ O(n^{-\frac{1}{2}})
	\end{equation}
	{\bf Step 2.} We apply the same strategy to the right tail. This time we compute
	\[
	\frac{S_t(n,i+1)}{S_t(n,i)}=\frac{(1-s_t) (n-i)}{(i+1) s_t}
	\]
	we define
	\[
	M_{\lambda,t,n} = \left\lceil \frac{\lambda  n (1-s_t)-s_t}{\lambda (1-s_t)+s_t} \right\rceil
	\]
	to ensure that
	\[
	i \geq M_{\lambda, t, n} \implies  \frac{S_t(n,i+1)}{S_t(n,i)} \leq \frac{1}{\lambda}.
	\]
	Repeating the same argument as for the left tail tells us that
	\begin{equation}
	\lambda_{n,t}:=\left( \sum_{i=N_{3,t,n}}^{M_{3,t,n}} \left( \frac{i}{n} \right)^{1-\alpha} S_t(n,i) z^{\alpha}_{i, t+\delta t} \right)
	+ O(n^{-\frac{1}{2}})
	\label{eq:lambdamiddle}
	\end{equation}
	We note that $\left(\frac{i}{n}\right)^{1-\alpha}$ is monotonic in
	$i$ and that  $\frac{N_{3,t,n}}{n}$ and $\frac{N_{3,t,n}}{n}$ tend
	to finite, non-zero limits as $n\to \infty$. We deduce that
	there exists a constant $C_{3,t}$ such that
	\begin{equation}
	N_{3,t,n}\leq i \leq M_{3,t,n} \implies \left|\left( \frac{i}{n} \right)^{1-\alpha} \right| \leq C_{3,t}
	\label{eq:summandBound}.
	\end{equation}
	This implies that
	\[
	\sum_{i=N_{3,t,n}}^{M_{3,t,n}} \left( \frac{i}{n} \right)^{1-\alpha} S_t(n,i)
	\leq C_{3,t} \sum_{i=N_{3,t,n}}^{M_{3,t,n}} S_t(n,i) \leq C_{3,t}
	\]
	Using  this together with our induction hypothesis,
	we may obtain from \eqref{eq:lambdamiddle} that
	\begin{equation}
	\lambda_{n,t}:=\left( \sum_{i=N_{3,t,n}}^{M_{3,t,n}} \left( \frac{i}{n} \right)^{1-\alpha} S_t(n,i) z^{\alpha}_{\infty, t+\delta t} \right)
	+ O(n^{-\frac{1}{2}}).
	\label{eq:lambdazinfinity}
	\end{equation}	
	{\bf Step 3.} In order to apply the bound \eqref{eq:bhattacharya},
	we define a Bernoulli random variable $X_{i,t}$ which takes the value $1$ if
	the $i$-th individual survives at time $t$ and $0$ otherwise. Thus
	$S_t(n,i)$ is the probability that $\sum_{j=1}^n X_{j,n}=i$.
	We define scaled random variables $\tilde{X}_{j,t}$ 
	of mean $0$ and standard deviation $1$ by
	\[
	\tilde{X}_{j,t} = \frac{X_j - s_t}{\sqrt{s_t(1-s_t)}},
	\]
	and so the appropriately scaled partial
	sum $Q_n$ is given by
	\[
	Q_n = \frac{1}{\sqrt{n}} \sum_{i=1}^n \tilde{X}_n = \frac{1}{\sqrt n}
	\sum_{j=1}^n \frac{X_j - s_t}{\sqrt{s_t(1-s_t)}}
	=
	\frac{ \left(\sum_{i=1}^n X_i\right) - n s_t}{\sqrt{n s_t(1-s_t)}}.
	\]
	We now wish to rewrite \eqref{eq:lambdazinfinity} as an integral.
	\[
	\lambda_{n,t} = 
	\int
	\id_{[N_{3,t,n},M_{3,t,n}]}(i) \left( \frac{i}{n} \right)^{1-\alpha} z^{\alpha}_{\infty, t+\delta t} 
	\ed (\sum_{j=1}^n X_{j,t})(i)
	+ O(n^{-\frac{1}{2}}).
	\]
	
	We make the substitution $i=n s_t + x \sqrt{n s_t(1-s_t)}$ to find
	\begin{multline}
	\lambda_{n,t} = 
	\int \Big[
	\id_{[N_{3,t,n},M_{3,t,n}]}(n s_t + x \sqrt{n s_t(1-s_t)}) \\
	\times \left( \frac{n s_t + x \sqrt{n s_t(1-s_t)}}{n} \right)^{1-\alpha} z^{\alpha}_{\infty, t+\delta t} \Big]
	\, \ed Q_n(x)
	+ O(n^{-\frac{1}{2}}).
	\end{multline}
	We may rewrite this as
	\begin{equation}
	\lambda_{n,t} = 
	\int \id_{[{\ell_{n,t}},{u_{n,t}}]}(x)
	\left( \frac{n s_t + x \sqrt{n s_t(1-s_t)}}{n} \right)^{1-\alpha} z^{\alpha}_{\infty, t+\delta t}
	\, \ed Q_n(x)
	+ O(n^{-\frac{1}{2}})
	\label{eq:lambdaasintegral}
	\end{equation}
	where
	\[
	\ell_{n,t}:= (N_{3,t,n} - n s)/\sqrt{n s(1-s)}, \quad
	u_{n,t}:= (M_{3,t,n} - n s)/\sqrt{n s(1-s)}.
	\]
	From our expressions for $N_{3,t,n}$ and $M_{3,t,n}$ one
	readily sees that $\ell_{n,t}$ tends to $-\infty$ at a rate proportional to $O(-\sqrt{n})$
	as $n\to\infty$. Likewise $u_{n,t}$ tends to $+\infty$ at a rate $O(\sqrt{n})$ as $n\to\infty$. We will assume that $n$ is large enough to ensure that $\ell_{n,t} <0 <u_{n,t}$.
	
	Let us define $g$ by
	\begin{equation}
	g = \id_{[{\ell_{n,t}},{u_{n,t}}]}(x)
	\left( \frac{n s_t + x \sqrt{n s_t(1-s_t)}}{n} \right)^{1-\alpha}.
	\end{equation}
	By \eqref{eq:summandBound}, $g$ is bounded by a constant independent of $n$.
	Hence $\omega_g(\R)$ is bounded independent of $n$. We can bound the derivative of $g$ inside the interval $({\ell_{n,t}},{u_{n,t}})$, independent of $n$. Hence for
	any $x \in (\frac{1}{2}{\ell_{n,t}},\frac{1}{2}{u_{n,t}})$ and
	for sufficiently large $n$, $\omega_g(B_{c n^{-\frac{1}{2}}}(x))< C_{4,t} n^{-\frac{1}{2}}$
	for some constant $C_{4,t}$ independent of $n$. It follows that
	\begin{equation}
	\int \id_{[\frac{1}{2}{\ell_{n,t}},\frac{1}{2}{u_{n,t}}]}(x) \omega_g(B_{c n^{-\frac{1}{2}}}(x)) \, \ed \Phi(x) = O(n^{-\frac{1}{2}}).
	\label{eq:oscbound1}
	\end{equation}
	Since $\ell_{n,t}$ tends to $-\infty$ at a rate proportional to $O(-\sqrt{n})$, since $g$ is bounded, and since the normal distribution
	has super-exponential decay in the tails, we have
	\begin{equation}
	\int \id_{(-\infty,\ell_{n,t}]} \omega_g(B_{c n^{-\frac{1}{2}}}(x)) \, \ed \Phi(x) = O(n^{-\frac{1}{2}})
	\label{eq:oscbound2}
	\end{equation}
	and similarly
	\begin{equation}
	\int \id_{[u_{n,t},\infty)} \omega_g(B_{c n^{-\frac{1}{2}}}(x)) \, \ed \Phi(x) = O(n^{-\frac{1}{2}}).
	\label{eq:oscbound3}
	\end{equation}
	Estimates \eqref{eq:oscbound1}, \eqref{eq:oscbound2}, \eqref{eq:oscbound3}
	together with the bound on $\omega_g(\R)$ allow us to apply the Central Limit Theorem estimate \eqref{eq:bhattacharya} to \eqref{eq:lambdaasintegral}. We note that Cram\'er's condition holds. The result is
	\begin{equation*}
	\lambda_{n,t} = 
	\int \id_{[{\ell_{n,t}},{u_{n,t}}]}(x)
	\left( \frac{n s_t + x \sqrt{n s_t(1-s_t)}}{n} \right)^{1-\alpha} z^{\alpha}_{\infty, t+\delta t}
	\, \ed \Phi(x)
	+ O(n^{-\frac{1}{2}})
	\end{equation*}
	
	We now apply Laplace's method to estimate this Gaussian
	integral (see Proposition 2.1, page 323 in \cite{ss03} ) and obtain
	\begin{equation*}
	\lambda_{n,t} = 
	s_t^{1-\alpha}
	+ O(n^{-\frac{1}{2}})
	\label{eq:lambdaasgaussian}
	\end{equation*}
	From the definition of $\tilde{\theta}$ in equation \eqref{eq:deftildetheta} and our definition of $\lambda_{n,t}$
	in \eqref{eq:defnlambdant} we obtain
	\begin{equation*}
	\tilde{\theta}_{n,t} =  \beta^\frac{1}{\rho} \exp( \xi \, \delta t) 
	s_t^\frac{{1-\alpha}}{\alpha} z_{\infty,t+\delta t} + O(n^{-\frac{1}{2}}).
	\end{equation*}
	We may now compare this with the definition of $\theta_t$ given in \eqref{eq:defphi} for the infinitely collectivised case $C=1$.
	We see that in this case
	\begin{equation*}
	\tilde{\theta}_{n,t} = \theta_t + O(n^{-\frac{1}{2}}).
	\end{equation*}
	It now follows from the recursion relations for $z_{n,t}$
	and $z_{\infty,t}$ (given in \eqref{eq:valueFunction4}
	and \eqref{eq:valueFunction4Tilde} respectively) together with our induction hypothesis
	that 
	\begin{equation*}
	z_{n,t} = z_{\infty,t} + O(n^{-\frac{1}{2}}).
	\end{equation*}
	
	This completes the induction step and the proof.
\end{proof}

\section{Systematic Longevity Risk}

\subsection{Justification for continuous-time aggregator}
\label{sec:ezAggregatorMotivation}

If an individual's time of death is independent of the systematic factors, their discrete-time Epstein-Zin utility with mortality
satisfies
\[
Z_{t}^\rho=c_t^\rho + e^{-(\delta + \frac{\rho}{\alpha} \lambda_t) \delta t} \E_{\P}( Z_{t+\delta t}^\alpha  \mid {\cal F}_t)^{\frac{\rho}{\alpha}}. 
\]
where we have introduced a discounting rate $\delta$ so that $\beta=e^{-\delta \, \delta t}$.
Rearranging we find
\[
\E_{\P}( Z_{t+\delta t}^\alpha  \mid {\cal F}_t) = \left[
\frac{Z_{t}^\rho- c_t^\rho}
{e^{-(\delta + \frac{\rho}{\alpha} \lambda_t) \delta t}}
\right]^\frac{\alpha}{\rho}. 
\]
We define $V$ by requiring $\alpha V=Z^\alpha$. The coefficient $\alpha$ is there to ensure that the transform $Z \to V$ is monotone so that $V$ is a gain
function that defines identical preferences to $Z$. So
\[
\E_\P( V_{t+\delta t} \mid {\cal F}_t )
= \frac{1}{\alpha}
\left[
\frac{
(\alpha V)^{\frac{\rho}{\alpha}} - \delta t \, c^\rho
}
{e^{-(\delta + \lambda \frac{\rho}{\alpha})\Delta t }}
\right]^{\frac{\alpha}{\rho}}
\]
Proceeding formally, we use l'H\^opital's rule to find:
\begin{align}
\frac{\ed}{\ed s} \E_\P( V_{s} \mid {\cal F}_t )
\Big|_{s=t}
&=\lim_{\delta t \to 0} \frac{\ed}{\ed (\delta t) }
\left[
\frac{1}{\alpha}
\left[
\frac{
(\alpha V)^{\frac{\rho}{\alpha}} - \delta t \, c^\rho
}
{e^{-(\delta + \lambda_t \frac{\rho}{\alpha})\delta t }}
\right]^{\frac{\alpha}{\rho}}
\right] \nonumber \\
&=\frac{1}{\alpha} \frac{\alpha}{\rho} (- c^\rho)((\alpha V)^\frac{\rho}{\alpha} )^{\frac{\alpha}{\rho}-1}
+ \left( \frac{\alpha}{\rho} \delta + \lambda_t \right) V
\nonumber \\
&=-\frac{1}{\rho}c^\rho \left( \alpha V \right)^{1-\frac{\rho}{\alpha}} + \left( \frac{\alpha}{\rho} \delta + \lambda_t \right) V.
\label{eq:ezAggregatorMotivation}
\end{align}
This then motivates the definition for the Epstein--Zin aggregator with mortality as a solution to the BSDE \eqref{eq:ezBSDE}
will satisfy equation \eqref{eq:ezAggregatorMotivation}.

\subsection{Computation of the HJB equation}

We assume the value function $V(t,w,\lambda,S)$ is smooth
and apply It\^o's lemma to obtain:
\[
\begin{split}
\ed V &=
\Bigg[
\frac{\partial V}{\partial t} + \frac{ \partial^2 V }
{\partial w \partial S} q(\sigma S)^2
+ \frac{1}{2} \frac{\partial^2 V}{\partial w^2}(q\sigma S)^2 + \frac{1}{2} \frac{\partial^2 V}{\partial S^2} (\sigma S)^2 +  \drift_\lambda\frac{\partial V}{\partial \lambda}\\
&\quad  + \frac{\partial V}{\partial w} \left(\lambda  w k + r(w-qS)-c+q\mu S\right)
+ \frac{\partial V}{\partial S} \mu S + \frac{1}{2} \vol_\lambda^2 \frac{\partial^2 V}{\partial \lambda^2}
\Bigg] \, \ed t \\
&\quad
+ q\sigma S \ed W^1_t + \vol_\lambda\frac{\partial V}{\partial \lambda}d W^2_t .
\end{split}
\]
The dynamic programming principle tells us $h:=V +\int_0^tf(c,V,\lambda)\ed s $ is a martingale for the optimal strategy and a supermartingale otherwise. Hence, we require the drift of $h$ to equal zero for the optimal strategy $(c^*,q^*)$ and be less than zero otherwise. Thus, our HJB equation is 

\begin{align}
    0=\underset{(c,q)\in\mathcal{A}(w_0)}{\sup}&\left[\frac{\partial V}{\partial t} + \frac{ \partial^2 V }
    {\partial w \partial S} q(\sigma S)^2
    + \frac{1}{2} \frac{\partial^2 V}{\partial w^2}(q\sigma S)^2 + \frac{1}{2} \frac{\partial^2 V}{\partial S^2} (\sigma S)^2  \right.\nonumber\\
    &\left.+  \drift_\lambda\frac{\partial V}{\partial \lambda}
    + \frac{\partial V}{\partial w} \left(\lambda w k + r(w-qS)-c+q\mu S\right)\right.
    \nonumber\\
    &\left.+ \frac{\partial V}{\partial S} \mu S + \frac{1}{2} \vol_\lambda^2 \frac{\partial^2 V}{\partial \lambda^2}+ f(c_t,V_t,\lambda_t)\right].\label{eq:general_HJB_1_fund}
\end{align}
Differentiating with respect to $c$ and $q$, we obtain the following optimal consumption and investment strategies
\begin{align}
    & c^*= \left((\alpha V)^{\frac{\rho}{\alpha}-1}\frac{\partial V}{\partial w}\right)^{\frac{1}{\rho-1}} \label{eq:op_consumption}\\
    & q^*=\frac{(r-\mu) \frac{\partial V}{\partial w} - S \sigma^2 \frac{\partial^2 V}{\partial w\partial S}}{S \sigma^2 \frac{\partial^2 V}{\partial w^2}}\label{eq:op_quantity}.
\end{align}
Substituting the ansatz \eqref{eq:ansatz}, \eqref{eq:op_consumption} and \eqref{eq:op_quantity} into \eqref{eq:general_HJB_1_fund}, yields the equation \eqref{eq:g_eqtn}.

\subsection{The supremum of the value function for ill-posed problems}

\label{sec:complexSolutions}

We now explain why our non-trivial analytical solution is complex for certain values of $\alpha$ and $\rho$, and what the value function must be in these cases, at least for $a=b=0$. First, we return to our most general HJB equation \eqref{eq:general_HJB_1_fund}. Here, we now take $r=\mu=0$ (which implies the quantity of the risky asset purchased is $q=0$), $a=b=0$ and let the consumption rate $c=\Tilde{c}w$ be proportional to wealth (as we should still have a symmetry of the HJB equation with respect to wealth), but otherwise unspecified, so that the supremum is not necessarily obtained. As for our ansatz, we let 
\begin{equation}
    V(\tilde{c},\lambda)=w^\alpha g(\Tilde{c},\lambda).
\end{equation}
The resulting HJB equation is the following algebraic expression
\begin{equation}
    g(\tilde{c}, \lambda ) \left((k\alpha -1) \lambda +\frac{\alpha  c^{\rho } (\alpha  g(\tilde{c}, \lambda ))^{-\frac{\rho }{\alpha }}}{\rho }-\alpha  c\right)=0,\label{eq:c_hjb_eqtn}
\end{equation}
which has the trivial solution $g(\tilde{c},\lambda)=0$ and
\begin{equation}
    g(\tilde{c},\lambda)=\frac{c^{\alpha } \left(\frac{\rho  (\lambda(1-k\alpha) +\alpha  \tilde{c} )}{\alpha }\right)^{-\frac{\alpha }{\rho }}}{\alpha }. \label{eq:g_c_sol}
\end{equation}
When $k=0$ and the longevity credit system is removed, this solution is always real. We now focus on the case $k=1$ with longevity credits, which can produce complex solutions. In order to get a real non-trivial value function, we must require
\begin{equation}
    \tilde{c}\rho\geq \frac{\lambda(\alpha-1)\rho}{\alpha}:=\tilde{c}_c.\label{eq:c_condition}
\end{equation}
Checking this condition for all four sign combinations of $\alpha$ and $\rho$, we learn the following:
\begin{enumerate}
    \item In the top right quadrant, when $\alpha,\rho>0$, $c\geq c_c<0$, so our consumption is unbounded. As $c\to0^+$, we see $g\to0^+$ and hence $V\to0^+$ (at least for finite wealth). As $c\to\infty$ we see $g\to0^+$ and hence $V\to0^+$ (at least for finite wealth). A maximum can therefore be obtained and our analytic solution is meaningful.  
    
    \item In the top left quadrant, when $\alpha<0,\rho>0$, $c\geq c_c>0$, so our consumption is bounded from below. As $c\to\infty$, we see $g\to-\infty$ and hence $V\to-\infty$ (at least for finite wealth). As $c\to\frac{\lambda(\alpha-1)}{\alpha}^+$ we see $g\to0^+$ and hence $V\to0^+$ (at least for finite wealth). The best we can do is obtain a value function of $0$, which is possible since $g(\tilde{c},\lambda)=0$ is a solution to \eqref{eq:c_hjb_eqtn} in this case. 
    
    \item In the bottom left quadrant, when $\alpha<0,\rho<0$, $c\leq c_c>0$, so our consumption is bounded from above. As $c\to0^+$, we see $g\to-\infty$ and hence $V\to-\infty$ (at least for finite wealth). As $c\to\frac{\lambda(\alpha-1)}{\alpha}^-$ we see $g\to-\infty$ and hence $V\to-\infty$ (at least for finite wealth). A maximum can therefore be obtained and our analytic solution is meaningful. 
    
    \item In the bottom right quadrant, when $\alpha>0,\rho<0$, $c\leq c_c<0$, so there are no valid consumption strategies that yield a solution the form \eqref{eq:g_c_sol}. Hence, we only have the trivial solution $g(\tilde{c},\lambda)=0$, which satisfies \eqref{eq:c_hjb_eqtn} in this case.
\end{enumerate}

When $a$ and $b$ are non-zero, this picture is similar, but somewhat perturbed. For the specific case $a=4,b=1$, we find non-trivial regions in the $(\alpha,\rho)$-plane where our analytical solution is complex (see \Cref{fig:costOfSystematicMortalityRisk})

\subsection{Continuous time CBD mortality model details}\label{sec:CBD_details}

We consider a continuous time version of the two-factor Cairns-Blake-Dowd model, for which the underlying mortality effects are described by:
\begin{align}
    dA := d(A_1,A_2)= \mu dt + C dW_t,
\end{align}
where 
\begin{align*}
    \mu&:=(\mu_1,\mu_2) = (-0.00669, 0.000590) \text{ and }\\ 
    V&=CC^T=\begin{pmatrix}
            0.00611 & -0.0000939\\ 
            -0.0000939 & 0.000001509
            \end{pmatrix}.
\end{align*}
The parameter values used come from \cite{CBD_article} equation (4).
C is the upper triangular matrix from the Cholesky decomposition of V i.e., 
\begin{equation}
    C=
    \begin{pmatrix}
            0.0782 & -0.00120\\ 
            0 & 0.000257            
            \end{pmatrix}.\label{eq:vol_matrix}
\end{equation}
The process $A_1$ captures general improvements in mortality over time at all ages (and therefore trends downwards with time), while $A_2$ captures the fact mortality  improvements have been greater at lower ages (and therefore trends upwards with time) \cite{CBD_article}. 

We next define $p(x_0,t)$ to be the survival probability at time $t$, for a cohort of age $x_0$ at time $t=0$, and $$q(x,t)=1-p(x,t)=\frac{e^{A_1 +A_2 (x_0+t)}}{1+e^{A_1 +A_2 (x_0+t)}}$$ to be the probability of death at time $t$. We assume $q(x_0,t)=1-e^{-\lambda(x_0,t)}$, with $\lambda(x_0,t)$ essentially being a time dependent rate parameter for an exponential distribution. We therefore interpret $\lambda(x_0,t)$ as our mortality rate and have 
\begin{equation}  
\lambda(x_0,t) = -\text{log}\left(1-\frac{e^{A_1 +A_2 (x_0+t)}}{1+e^{A_1 +A_2 (x_0+t)}}\right).\label{eq:rate_def}
\end{equation}

We would next like to derive the SDE for $\lambda$, which is done by an application of It\^{o}'s Lemma. In its current form, we obtain an SDE with a stochastic drift due to the $A_2 (x_0+t)$ term. This increases the dimension of the HJB equation from 2 to 3. For simplicity, we make the following modification. Note that the SDE for $A_1+A_2(x_0+t)$, is given by
\begin{align}
    d(A_1&+A_2(x_0+t)) \nonumber\\
    &= dA_1 + dA_2 (x_0+t) + A_2 dt\nonumber\\
    &= (\mu_1 +\mu_2(x_0+t) +A_2)dt+C_{1,1}dW^1_t+C_{1,2}dW^2_t+(x_0+t)C_{2,2}dW^2_t\nonumber\\
    &=(\mu_1 +\mu_2(x_0+t) +A_2)dt+\left(C_{1,1}^2+C_{1,2}^2+((x_0+t)C_{2,2})^2\right)^\frac{1}{2} d\Tilde{W}_t,\label{eq:underlying_SDE}
\end{align}
where $C_{i,j}$ denotes the $i,j$'th component of \eqref{eq:vol_matrix} and $\Tilde{W}$ is a new independent Brownian motion. To obtain a tractable problem, we define a new variable $x:=A_1+A_2(x_0+t)$ in \eqref{eq:rate_def}, which satisfies the SDE
\begin{equation}
    dx = (\mu_1 +\mu_2(x_0+t) +A_{2,0}+\mu_2 t)dt+\left(C_{1,1}^2+C_{1,2}^2+((x_0+t)C_{2,2})^2\right)^\frac{1}{2} d\Tilde{W}_t,\label{eq:simple}
\end{equation}
i.e., \eqref{eq:underlying_SDE} with $A_2$ replaced by the solution of the ODE $dA_{2}=\mu_2 dt$. This then yields a HJB equation dependent on $t$ and $\lambda$ only. Looking at Figure 1 in \cite{CBD_article}, we choose $A_2^0=0.1058$. 

The impact of this simplification (i.e. the difference between simulating \eqref{eq:underlying_SDE} and \eqref{eq:simple}) is a lower mortality rate at older ages. Looking at \Cref{fig:1factor_vs_2factor}, there is little difference in the mortality rates up to age 80. The difference only becomes notable from age 90 onwards, by which point the prediction of either model becomes less meaningful. 

\begin{figure}
    \centering
        \includegraphics[width=0.8\textwidth]{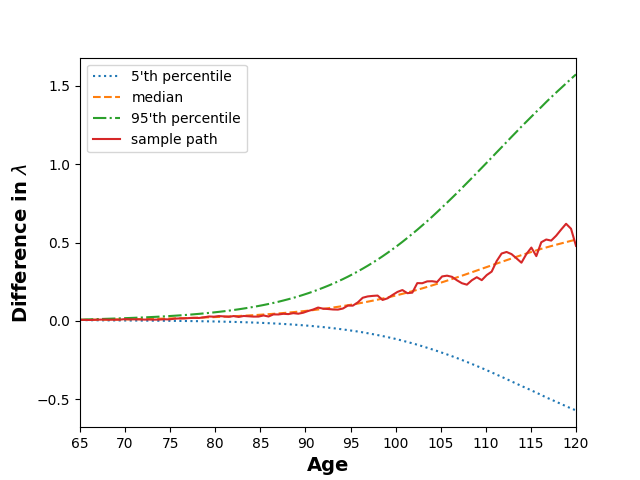}
        \caption{Fan diagram of the difference in simulated mortality rates from \eqref{eq:underlying_SDE} $\lambda_1$ and \eqref{eq:simple} $\lambda_2$ i.e, $\lambda_1-\lambda_2$. Here, we consider $10^6$ scenarios and take 100 equally spaced time steps.}
    \label{fig:1factor_vs_2factor}
\end{figure}

\subsection{Numerical approach for solving the CBD problem }\label{sec:numerical_approach}

We solve the HJB equation for the CBD model using the Crank-Nicolson scheme, along with an improved Euler step to generate the first estimate of the solution at each time step.
We set $x_0=65$ to be the starting age of our simulations. $\lambda=0.01$ corresponds to a survival probability of 0.99 at retirement age, while $\lambda=10$ yields a survival probability of $4.5\times 10^{-5}$, so that all individuals are almost surely dead at this point. We are therefore interested in the solution for $\lambda\in[0.01,10]$.

The value of the solution at the boundaries is not obvious a priori. As such, we enforce Neumann boundary conditions at the ends of our domain and extend our computational domain to $\lambda\in[0.001,20]$ to minimise the impact of the boundary conditions on the solution for $\lambda\in[0.01,10]$, the points of interest. We expect $g\to\infty$ as $\lambda\to0$ and $g\to0$ as $\lambda\to\infty$. Because of this expected behaviour, we take advantage of two transformations: (i) we use the change of variables $L=\textrm{log}(\lambda)$ and (ii) solve the HJB equation for $\textrm{log}(g)$. We first solve the riskless equation as it only requires one boundary condition plus the payoff. Using the stylised model as a guide, we expect $g\approx \lambda^{\frac{\alpha (\rho-1)}{\rho}}$ to be a rough approximation to the solution. Taking into account our transformations, we therefore enforce $$\frac{\partial}{\partial L}  (\textrm{log} (g))= \frac{\alpha(\rho-1)}{\rho}.$$ With the risk-less solution $g_{det}$ obtained, we solve the full HJB equation with $$\frac{\partial}{\partial L}  (\textrm{log} (g(20)))=\frac{\partial}{\partial L}  (\textrm{log} (g_{det}(20)))$$ and $$\frac{\partial}{\partial L}  (\textrm{log} (g(0.001)))=\frac{\partial}{\partial L}  (\textrm{log} (g_{det}(0.001))),$$ where the derivatives are approximated using backward and forward difference approximations respectively.

Again, informed by \eqref{eq:analytical} from the stylised model, we take our payoff at the final time $T_f$, to be $$g_{T_f}=\lambda^{\frac{\alpha (\rho-1)}{\rho}},$$
which should be a rough approximation to the correct payoff. Since this is an approximation, we would like to minimise its impact on the solution. We therefore run our numerical scheme back from time $T_f=150$ (which corresponds to age 215). Our computational domain for $\lambda$ extends to $\lambda=20$, and in a deterministic setting (i.e., the volatility is zero in \eqref{eq:CBD_lambda}) the simulated value of $\lambda$ is approximately 20 at age 215, hence the choice. This time is far longer than is required for all individuals to die (our simulations suggest everyone will be dead by age 120 (see \Cref{fig:CBD}) (b)) and therefore achieves our aim of minimising the payoff's impact

We would also like to generate typical wealth and consumption streams for a fund that behaves optimally. This can be done as follows. Having obtained the numerical solution $g_{num}(\lambda)$ to \eqref{eq:g_eqtn}, it is not difficult to find constants $A,B$ such that $g_{num}\approx A \lambda^B$. This can then be used to approximate $c^*$ in \eqref{eq:op_consumption}. Along with 
$q^*$ given by \eqref{eq:op_quantity}, this can be feed into \eqref{eq:wealth_eq}. We can then simulate this approximation to \eqref{eq:wealth_eq} using the Euler-Maruyama method \cite{intro_book}.


\end{document}